\def\notes{1}
\documentclass[11pt]{article}

\usepackage{amsmath, amsthm, amssymb, graphicx, enumerate, fullpage, url}
\usepackage{color}

\newcommand{\mnote}[1]{\ifnum\notes=1{{\sf\color{red} [Madhu: #1]}}\fi}
\newcommand{\snote}[1]{\ifnum\notes=1{{\sf\color{blue} [Swastik: #1]}}\fi}

\theoremstyle{plain}% default
\newtheorem{thm}{Theorem}[section]
\newtheorem{theorem}[thm]{Theorem}
\newtheorem{definition}[thm]{Definition}
\newtheorem{proposition}[thm]{Proposition}

\newtheorem{prop}[thm]{Proposition}
\newtheorem{lemma}[thm]{Lemma}

\newtheorem{claim}[thm]{Claim}

\theoremstyle{definition}

\newtheorem{example}[thm]{Example}
\mathchardef\mhyphen="2D
\newcommand{\cF}{\mathcal{F}}
\newcommand{\E}{\mathbb{E}}
\newcommand{\calf}{\mathcal{F}}
\newcommand{\calL}{\mathcal{L}}

\newcommand{\F}{\mathbb{F}}
\newcommand{\K}{\mathbb{K}}

\newcommand{\wt}{{\rm wt}}

\newcommand{\modstar}[1]{~(\mathop{\rm{mod}^*} #1)}
\newcommand{\Z}{\mathbb{Z}}

\DeclareMathOperator{\Deg}{Deg}
\DeclareMathOperator{\Fam}{Fam}

\DeclareMathOperator{\Tr}{Tr}

\DeclareMathOperator{\supp}{supp}

\newcommand{\lift}{\mathop{\mathrm{Lift}}}
\newcommand{\Corr}{\mbox{{\sc Corr}}}
\newcommand{\Tester}{\mbox{{\sc Test}}}
\newcommand{\ip}[1]{\langle #1 \rangle}
\newcommand{\Lift}{\mathop{\mathrm{Lift}}}
\newcommand{\poly}{\mathop{\mathrm{poly}}}
\newcommand{\aff}{\mathop{\mathrm{Aff}}}

\renewcommand{\vec}[1]{{\mathbf #1}}
\newcommand{\bA}{{\mathbf A}}
\newcommand{\bB}{{\mathbf B}}

\begin{document}
%%%%%%%%%%%%%%%%%%%%%%%%%%%%%%%%%%%%%%%%%%%%

\title{New affine-invariant codes from lifting}

\author{Alan Guo\thanks{CSAIL, Massachusetts Institute of
Technology, 32 Vassar Street, Cambridge, MA, USA. {\tt aguo@mit.edu}. Research
supported in part by NSF grants CCF-0829672, CCF-1065125,
CCF-6922462, and an NSF Graduate Research Fellowship.}
\and {Swastik Kopparty \thanks{Department of Computer Science \& Department of Mathematics, Rutgers University, Piscataway NJ, USA. {\tt swastik.kopparty@rutgers.edu}.}}
\and Madhu Sudan\thanks{Microsoft Research New England,
One Memorial Drive, Cambridge, MA 02139, USA.
 {\tt madhu@mit.edu}.}
}

\maketitle

\begin{abstract}
In this work we explore error-correcting codes derived from
the ``lifting'' of ``affine-invariant'' codes.
Affine-invariant codes are simply linear codes whose coordinates
are a vector space over a field and which are invariant under
affine-transformations of the coordinate space. Lifting takes codes
defined over a vector space of small dimension and lifts them to higher
dimensions by requiring their restriction to every subspace of
the original dimension to be a codeword of the code being lifted.
While the operation is of interest on its own, this work focusses
on new ranges of parameters that can be obtained by such codes,
in the context of local correction and testing.
In particular we present four interesting ranges of parameters that
can be achieved by such lifts, all of which are new in the context
of affine-invariance and some may be new even in general.
The main highlight is a construction of high-rate codes
with sublinear time decoding. The only prior construction of such
codes is due to Kopparty, Saraf and Yekhanin~\cite{KSY}.
All our codes are extremely simple, being just lifts of various
parity check codes (codes with one symbol of redundancy), and
in the final case, the lift of a Reed-Solomon code.

We also present a simple connection between certain lifted codes 
and lower bounds on the size of ``Nikodym sets''. Roughly, a Nikodym set
in $\F_q^m$ is a set $S$ with the property that every point has a line
passing through it which is almost entirely contained in $S$. While
previous lower bounds on Nikodym sets were roughly growing as
$q^m/2^m$, we use our lifted codes to prove a lower bound of $(1 - o(1))q^m$
for fields of constant characteristic.
\end{abstract}

\newpage

\section{Introduction}

In this work we explore the ``locality properties'' of some highly symmetric
codes constructed by ``lifting'' ``affine-invariant'' codes.
We describe these terms below.

\subsection{Basic terminology and background}

We start with some standard coding theory preliminaries.
Let $\F_q$ denote the finite field of cardinality $q$ and
for any finite set $D$, let $\{D \to \F_q\}$ denote the
set of all functions from $D$ to $\F_q$. In this work, a
code on coordinate set $D$
is a set of functions $\calf \subseteq \{D \to \F_q\}$.
A code $\calf$ is said to be linear if it forms a vector space over
$\F_q$, i.e., if for every $f,g \in \calf$ and $\alpha \in \F_q$
the function $\alpha f + g \in \calf$.
We refer to $N = |D|$ as the length of the code.
A second parameter of interest is the dimension of the code
which is the dimension of $\calf$ as a vector space.
The dual of a code $\calf$, denoted $\calf^\bot$,
is the set of functions $\{g:D \to \F_q | \ip{f,g} = 0 \text{ } \forall f \in
\calf\}$, where $\ip{f,g} = \sum_{x \in D} f(x)g(x)$ denotes the
standard inner product of vectors.
Let $\wt(f) = |\{x \in D | f(x) \ne 0\}|$ denote the weight
of $f$.
Let $\delta(f,g) = |\{ x \in D | f(x) \ne g(x) \}|/|D|$ denote
the (normalized Hamming) distance between $f$ and $g$.
(So $\delta(f,g) = \wt(f-g)/|D|$.)
We say $f$ is $\delta$-close to $g$ if $\delta(f,g) \leq \delta$
and $\delta$-far otherwise. We say $f$ is $\delta$-close
to $\calf$ if there exists $g \in \calf$ that is $\delta$-close
to $f$ and $\delta$-far otherwise.
We say $\calf$ is a code of distance $\delta$ if every pair of
distinct codewords in $\calf$ are $\delta$-far from each other.
We use $\delta(\calf)$ to denote the maximum $\delta$ such that $\calf$
is a code of distance $\delta$.

In this work we explore some aspects of affine-invariant codes. In such
codes
the domain $D$ is a vector space $\F_{q^n}^m$, i.e., an $m$-dimensional
vector space over the $n$-dimensional extension field of the
range $\F_q$. Let $Q = q^n$ and let $\F_Q$ denote the field
of size $Q$.
We say a function $A:\F_Q^m \to \F_Q^m$ is an affine function
if $A(x) = M\cdot x + b$ for some matrix $M \in \F_Q^{m \times m}$
and vector $b \in \F_Q^m$. We say $A$ is an affine permutation
if $M$ is invertible.
A code $\calf \subseteq \{\F_{Q}^m \to \F_q\}$ is said to be
affine-invariant if for every affine permutation
function $A:\F_Q^m \to \F_Q^m$ and for every $f \in \calf$
the function $f \circ A$ given by $(f \circ A)(x) = f(A(x))$
is also in $\calf$.\footnote{In some of the earlier works
invariance is defined with respect to all affine functions and
not just permutations. In Section~\ref{app:perm} we show that the
two notions are equivalent and so we use invariance with respect
to permutations in this paper.}

Affine-invariant codes are of interest to us because they exhibit,
under natural and almost necessary conditions, very good locality
properties: they tend to be locally testable and locally correctible.
We introduce these notions below.
We say a code $\calf$ is $(k,\delta)$-locally correctible
($(k,\delta)$-LCC)
if there exists a probabilistic algorithm $\Corr$ that,
given $x \in D$ and
oracle access to a function $f:D \to \F_q$ which is $\delta$-close
to some $g \in \calf$, makes at most $k$ queries to $f$ and
outputs $g(x)$ with probability at least $2/3$.
We say that $\calf$ is $(k,\epsilon,\delta)$-locally testable
($(k,\epsilon,\delta)$-LTC)
if $\calf$ is a code of distance $\delta$ and there exists a
probabilistic algorithm $\Tester$ that, given oracle access to
$f:D \to \F_q$, makes at most $k$ queries to $f$ and accepts
$f \in \calf$ with probability one, while rejecting $f$ that
is $\tau$-far from $\calf$ with
probability at least $\epsilon \cdot \tau$.

\subsection{This work: Motivation and Results}

As noted above affine-invariant lead naturally to locally decodable
codes and locally testable codes. In this work we use a certain
lifting operation to exhibit codes with very good locality. We
start by defining the lifting operation. For a function
$f:\F_Q^m \to \F_q$ and set $S \subseteq
\F_Q^m$ let $f|_S$ denote the restriction of $f$ to the domain $S$.

\iffalse{
\begin{definition}[Lifting]
For $t \leq m$,
let $\aff_t(\F_Q^m)$ denote the set of all $t$-dimensional
affine subspaces of $\F_Q^m$. For $f:\F_Q^m \to \F_q$,
and $V \in \aff_t(\F_Q^m)$,
let $f|_V:\F_Q^t \to \F_Q$ denote the restriction of $f$ to
the subspace $V$ (under some canonical affineness preserving
isomorphism between $\F_Q^t$ and $V$).
For a code $\calf \subseteq \{\F_Q^t \to \F_q\}$, and integer
$m \geq t$ its $m$-dimensional lift
$\Lift_m(\calf) \subseteq \{\F_Q^m \to \F_q\}$ is the
code
$$\{f : \F_Q^m \to \F_q | f|_V \in \calf \forall V \in \aff_t(\F_Q^m)\}.$$
\end{definition}
}\fi

\begin{definition}[Lifting]
For a code $\calf \subseteq \{\F_Q^t \to \F_q\}$, and integer
$m \geq t$ its $m$-dimensional lift
$\Lift_m(\calf) \subseteq \{\F_Q^m \to \F_q\}$ is the
code
$$\{f : \F_Q^m \to \F_q ~|~ f|_V \in \calf \mbox{ for every
$t$-dimensional affine subspace $V \subseteq \F_Q^m$}\}.$$
\end{definition}

(Note that the definition above assumes some canonical way to equate
$t$-dimensional subspaces of $\F_Q^m$ with $\F_Q^t$. But for
affine-invariant families $\calf$ the exact correspondence does not
matter as long as the map is an isomorphism.)

The lift is a very natural operation on affine-invariant codes,
and builds long codes from shorter ones. Indeed, lifts may be
interpreted as the basic operation that leads to the construction
of ``(Generalized) Reed-Muller'' codes, codes formed by
$m$-variate polynomials
over $\F_q$ of total degree at most $d$: Such codes are the
``lifts'' of $t$-variate polynomials of degree at most
$d$, for $t = \lceil \frac{d+1}{q - q/p} \rceil$ where $p$
is the characteristic of $q$. (This follows from the ``characterization''
of polynomials as proven in \cite{KR06}.)
While the locality properties (testability and correctability)
of Reed-Muller codes are
well-studied~\cite{RS,ALMSS,AroraSudan,RazSaf,AKKLR,KR06,JPRZ,BKSSZ,HSS11},
they are essentially the only rich class of symmetric codes
that are well-studied.
The only other basic class of symmetric codes that are studied
seem to be sparse ones, i.e., ones with few codewords.

In this work we explore the lifting of codes as a means to building
rich new classes of {\em dense} symmetric codes.
(In Theorems~\ref{thm:one}
-~\ref{thm:four} below we describe some of the codes we obtain this way,
and contrast them with known results.) Along the way we also initiate
a systematic study of lifts of codes. Lifts of codes were introduced
first in \cite{BMSS}, who explored it to prove negative results ---
specifically, to build ``symmetric LDPC codes'' that are not
testable. (Their
definition was more restrictive than ours, and also somewhat less
clean.) Our work
is the first to explore positive use of lifts.

We remark that all codes constructed by lifting have relative distance
of at least $Q^{-t}$ and are $(Q^t,Q^{-t}/3)$-LCC's and $(Q^t,
\Omega(Q^{-2t}),Q^{-t})$-LTC's. The local correctability follows directly from
their definition, while the local testability is a consequence of
the main result of \cite[Theorem~2.9]{KS08-ECCC}. (See also
Proposition~\ref{prop:liftlcc}.) This general feature suffices for
three of our code construction, while in the fourth case we have to
analyze the decodability a little
more carefully.

%\snote{NOT sure if this should be here or somewhere else (or nowhere)}
\paragraph{An example.}

Let $q$ be a power of $2$, let $d = (1-\delta) q$ and let us consider the lift of
the set of all univariate polynomials over $\F_q$ of degree at most $d$ to $\F_q^2$.
Explicitly, we mean the code $\mathcal F$ consisting of all functions $f : \F_q^2 \to \F_q$ such that the
restriction of $f$ to any line of $\F_q^2$ is a univariate polynomial of degree at most $d$.
$\mathcal F$ is an affine-invariant linear space.

By construction, it is clear that $\mathcal F$ has a lot of local structure; this leads
to a simple local-correction algorithm for $\mathcal F$ based on picking random lines
and performing noisy univariate polynomial interpolation (i.e., Reed-Solomon decoding).
We will show that in fact $\mathcal F$ also has large dimension (when $\delta$ is small).
This leads to a high-rate locally correctable code.

Which functions $f : \F_q^2 \to \F_q$ lie in $\mathcal F$? We will give an answer to this
question later in the paper, in terms of the polynomial representation $f(X,Y) = \sum_{0 \leq i, j < q} a_{ij} X^i Y^j$.
Here since we are interested in showing that $\dim \mathcal F$ is large, it will suffice
for us to show that there are many linearly independent elements in $\mathcal F$.
To do this, we will study when a monomial $g(X,Y) = X^i Y^j$ is in $\mathcal F$. 
Note that if we restrict $g$ to a line $\ell(T) = (\alpha_1 T + \alpha_0, \beta_1 T + \beta_0)$, we
get the function
$$g|_\ell (T)  = (\alpha_1 T + \alpha_0)^i (\beta_1 T + \beta_0)^j = \sum_{r \leq i} \sum_{s \leq j} \alpha_1^r \alpha_0^{i-r} \beta_1^{s} \beta_0^{j-s} {i \choose r}{j\choose s} T^{r+s}.$$
This function will equal a univariate polynomial of degree at most $d$ at all points of $\F_q$
if, when we reduce it mod $T^q - T$, we see no monomials of degree $> d$. Reducing the above polynomial
mod $T^q - T$ amounts to replacing $T^{r+s}$ in the above expression with $T^{r+s \modstar q}$ (where
$a \modstar q = 0$ if $a = 0$ and $a \modstar q = b \in \{1,\ldots,
q-1\}$ if $a \ne 0$ and $a = b \pmod{q-1}$). This will happen if
$i, j$ satisfy the following criterion: for every $r \leq i, s \leq j$, if ${i \choose r} \neq 0 \mod 2$ and
${j \choose s} \neq 0 \mod 2$, then $r + s \modstar q \leq d$. Via Lucas' theorem (which gives a characterization
of when ${a \choose b} = 0 \mod 2$, we deduce that the monomial $X^iY^j$ is in $\mathcal F$ if $(i,j)$ lies in the
set:
$$ S = \{ (i, j) \mid \forall r \leq_2 i , j \leq_2 s ,  r + s \modstar q \leq d \},$$
where $a\leq_2 b$ means that set of coordinates that equal $1$ in the binary representation of $a$ is a subset
of the set of coordinates that equal $1$ in the binary representation of $b$.
Finally, an analysis of the set $S$ shows that its size is $\geq (1 - \epsilon_\delta) \cdot q^2$, where
$\epsilon_\delta \to 0$ as $\delta \to 0$. Thus the dimension of $\mathcal F$ is at least $ (1 - \epsilon_\delta) \cdot q^2$.

We will formally treat this example in greater generality in a later section.
Before that, we will build up the theory of lifts of multivariate codes.
In Proposition~\ref{prop:fam} we will see that affine-invariant codes are
completely characterized by (and in fact spanned by) the monomials
in the code; thus the dimension of the code above exactly equals $|S|$.

\paragraph{The constructions.}

For simplicity most codes are described for the case of fields
of characteristic two,
while the construction does generalize to other fields.
(The main exception is in Theorem~\ref{thm:two} where the
code is later applied in other cases, so we describe the
more general result.)
The codes in the first three theorems below 
are obtained by the lifting of the parity-check
code.
By making appropriate choices of $Q$ and $t$ we get codes with
different locality (and distance).
The fourth code works over large fields only and is obtained by lifting
the Reed-Solomon code.

Our first code has constant locality $k$, for $k$ being a power of $2$.
If the length of the code is $N$ (in our setting $N = Q^m$), then
the code has dimension $\Omega_k((\log N)^{k})$.

\begin{theorem}
\label{thm:one}
For every positive integer $t$ and $k = 2^t$, there exists a constant
$c_k > 0$ such that for every positive integer $m$ and $N = 2^m$,
there exists a binary 
code of length $N$, dimension at least $c_k (\log N)^{k-2}$
which is a $(k-1,k^{-1}/3)$-LCC, and a $(k,\Omega(k^{-2}),k^{-1})$-LTC.
\end{theorem}

To contrast this with other known codes, essentially the only
symmetric binary code known in this regime is the Reed-Muller code,
which has dimension $\Omega((\log N)^{\log k})$ for locality $k$.
Thus our code has significantly greater dimension in this regime.
Our results are also asymptotically optimal for affine-invariant
codes, by a result of Ben-Sasson and Sudan~\cite{BS11} which shows
that any affine-invariant code with such local correctability or
testability must have dimension $(\log N)^{k + O(1)}$.

For local correctability, these codes asymptotically
match the performance of
best-known codes, which would be obtained by taking Generalized
Reed-Muller codes over a field of size roughly $k$ and then
composing it with some binary code. Our codes are simpler to
describe and the symmetry comes without any loss of parameters.
Furthermore, for really small constants, say $k=4$ or $k=8$, these
codes seem to be better than previously known locally correctible
codes.

Our next two codes consider relatively large locality (growing
with $N$). The advantage with these codes is that the redundancy
(the difference between the length and the dimension) grows exceedingly
slowly. The first of these two codes considers the setting where
the locality is $N^{\epsilon}$ for some positive (but tiny) $\epsilon$.
In such cases, we get codes of dimension $N - N^{1 - \epsilon'}$
where $\epsilon' > 0$ if $\epsilon > 0$. Thus the dimension is
extremely close to the length.

\begin{theorem}
\label{thm:two}
For every $\epsilon > 0$ and prime $p$,
there exists $\epsilon'>0$ such
that for infinitely many $N$,
there is a $p$-ary code of length $N$, 
dimension $N - N^{1 - \epsilon'}$
which is a $(N^{\epsilon},N^{-\epsilon}/3)$-LCC and
a $(N^{\epsilon},\Omega(N^{-2\epsilon}),N^{-\epsilon})$-LTC.
\end{theorem}

The codes from Theorem~\ref{thm:two} are not new. These codes, and
in particular their exact dimension are well-known in the
literature in combinatorics~\cite{CH:proj,Smith:proj}.
Their locality was first
noted by Yekhanin~\cite{Yekhanin:pc}
who noticed in particular that they are LCCs. Our main contribution
is to note that these are (naturally) obtained from lifts.
In the process we get that these are affine-invariant codes and
so are also LTCs, a fact that was not known before. Finally, our
bounds while cruder, give better asymptotic sense of the redundancy
of these codes (and in particular note that the redundancy is sublinear
in the code length).

We remark that these codes have very poor distance and very poor
error-correcting capability. However, in the context of applications
such as constructions of PCPs (probabilistically checkable proofs,
see e.g.,~\cite{AroraBarak}) one does not need distance or error-correction capability
per se. All one seems to need is the local correction and decoding
capability. So the theorem above motivates the search for extremely
efficient PCPs, where the difference between the length of the PCP
and the length of the classical proof is sublinear, while allowing
for sublinear query complexity. Such a result, if at all possible,
would really be transformative in the use of PCPs as a positive
concept. We also note that these codes play a useful 
role in giving lower bounds on the size of Nikodym sets --- we will
elaborate on this shortly.

Next, we consider codes
of locality $\Omega(N)$, so linear in the length of the code.
This range of parameters was motivated by the recent result of
Barak et al.~\cite{BGHMRS} who used such codes (with additional
properties that we are not yet able to prove) to build ``small-set
expanders'' with many ``large eigenvalues''. We won't describe
the application here, but instead turn to the parameters they
sought. They wanted codes of length $N$ with locality $\epsilon N$
and dimension $N - \poly(\log N)$. The codes they used were
Reed-Muller codes. By exploring lifts we are able to suggest
some alternate codes. These codes do have slightly better dimension,
though unfortunately, the improvement is not asymptotically significant
(and certainly not close to any known limits).
Nevertheless we report the codes below.

\begin{theorem}
\label{thm:three}
For every $\epsilon > 0$ and for infinitely many $N$,
there is a binary code of length $N$, dimension $N - O_{\epsilon}((\log N)^{\log
1/\epsilon})$,
which is a $(\epsilon N,\frac13 (\epsilon N)^{-1})$-LCC and
a $(\epsilon N,\Omega((\epsilon N)^{-2}),(\epsilon N)^{-1})$-LTC.
\end{theorem}

We note that
Barak et al. also require the codes to be ``absolutely testable'',
a strong notion of testability that we do not achieve in this work. 
Indeed, it is unclear if the codes as described above will turn out
to be absolutely testable. In followup work to ours, Haramaty et
al.~\cite{HRS12}, do show that some codes constructed by the above
principle (but not all) are absolutely testable. The dimensions
of their codes are somewhere between those of Barak et al. and those
from the above theorem (so are still of no asymptotic significance).

Finally, we describe the most interesting choice of parameters.
Our final code has locality $N^\delta$ for arbitrarily small
$\delta > 0$, while achieving dimension $(1-\epsilon)N$ for arbitrarily
small $\epsilon > 0$. While the dimension of this code is smaller
than that of the codes of Theorem~\ref{thm:three}, it corrects a
constant positive fraction of errors.

\begin{theorem}
\label{thm:four}
For every $\epsilon,\delta > 0$ there exists $\tau > 0$
such that for infinitely many $N$,
there is a $q$-code of length $N$ over $\F_Q$, for $Q \approx N^{\delta}$,
of dimension $(1 - \epsilon)N$
which is a $(N^{\delta},\tau)$-LCC, for some $q \approx N^{\delta}$.
\end{theorem}

Till 2010, no codes achieving such a range of parameters were known.
In particular no code was known that achieved dimension greater
than $N/2$ while achieving $o(N)$ locality to correct constant
fraction of errors. In 2010, Kopparty et al.~\cite{KSY} introduced
what they called the ``multiplicity codes'' which manage to
overcome the rate $1/2$ barrier. Other than their codes, no other
constructions were known that achieved the parameters of
Theorem~\ref{thm:four} and our construction provides the first
alternate. We remark that while qualitatively our theorem matches
theirs, the behavior of $\tau$ as a function of $\epsilon$ and
$\delta$ is much worse in our construction. Nevertheless for
concrete values of $N$, $\epsilon$ and $\delta$ our
construction actually seems to perform quite well.
Also, whereas in the basic codes of \cite{KSY} are over larger
alphabets than $N$, our codes are naturally over much smaller alphabets.
(Of course, one can always use concatenation to reduce alphabet
sizes, but such operations do result in a loss in concrete settings
of parameters.)

Theorems~\ref{thm:one}-\ref{thm:four} are proved in
Section~\ref{sec:codes}.
While each of the codes above may be of interest on their own,
the underlying phenomenon, of constructing codes with interesting
parameters by lifting shorter codes is an important one. Given
our belief that lifting is an important operation that deserves
study, we also do some systematic analysis of lifts. In particular
in this work we show that lifting of a base code essentially preserves
distance. This preservation is not exact and we give examples proving this
fact.

\paragraph{Bounds on the size of Nikodym sets.}

One of the applications of our results is to bounding, from below,
the size of ``Nikodym sets'' over finite fields (of small characteristic).
We define this concept before describing our results.

A set $N \subseteq \F_q^m$ is said to be a {\em Nikodym set} if
every point $x$ has a line passing through it such that all points of the line,
except possibly the point $x$ itself, are elements of $N$. More precisely,
$N$ is a Nikodym set if for every $x \in \F_q^m$ there exists
$y \in \F_q^m \setminus \{\vec 0\}$ such that 
$\{x + t y | t \in \F^*_q\} \subseteq N$.

Nikodym sets are closely related to ``Kakeya sets'' --- the latter
contain a line in every direction, while the former contain almost
all of a line through every point.
A lower bound for Kakeya sets was proved by Dvir \cite{Dvir} using the polynomial
method and further improved by using ``method of multiplicities '' by 
Saraf and Sudan \cite{SS08} and Dvir et al.~\cite{DKSS}.
Kakeya sets have seen applications connecting its study to the
study of randomness extractors,
esp. \cite{DvirShpilka07,DvirWigderson}.
Arguably Nikodym sets are about as natural in this connection as
Kakeya sets. 

Previous lower bounds on Kakeya sets
were typically also applicable to Nikodym sets and led to bounds
of the form $|N| \geq (1 - o(1)) q^m/2^m$ where the $o(1)$ term
goes to zero as $q \to \infty$\footnote{In the $m=2$ case, better bounds
are known for Nikodym sets~\cite{FLS, Li}.}. In particular previous lower bounds
failed to separate the growth of Nikodym sets from those of Kakeya
sets. In this work we present a simple connection (see
Proposition~\ref{prop:nikodym})
that shows that existence of (high-rate) affine-invariant codes
that are lifts of non-trivial univariate codes yield (large) lower bounds on
the size of Nikodym sets. Using this connection we significantly
improve the known lower bound on the size of Nikodym sets over fields of
constant characteristic. 

\begin{theorem}
\label{thm:nikodym}
For every prime $p$, and every integer $m$, there exists
$\epsilon = \epsilon(p,m) > 0$ such that for every finite
field $\F_q$ of characteristic $p$, if $N \subseteq \F_q^m$
is a Nikodym set, then $|N| \geq q^m - q^{(1-\epsilon)m}$.
In particular if $q \to \infty$, then $|N| \geq (1 - o(1))\cdot q^m$.
\end{theorem}

Thus whereas previous lower bounds on the size of Nikodym sets allowed
for the possibility that the density of the Nikodym sets vanishes as
$m$ grows,
ours show that Nikodym sets occupy almost all the space.
One way to view our results is that they abstract the polynomial 
method in a more general way, and thus lead to stronger lower bounds
(in some cases).

\paragraph{Previous work on affine-invariance.}
The study of invariance, and in particular affine-invariance,
in property testing was initiated by Kaufman and Sudan~\cite{KS08-ECCC}
and there have been many subsequent works
\cite{BS11,GKS08,GKS09,GoldKauf,KaufWig,KaufLub,BGMSS11-ECCC,BMSS,KL11,BRS,GS:sums}.
Most of the works, with the exceptions of \cite{KaufWig,KaufLub}, study
the broad class with the aim of characterizing all the testable properties.
The exceptions, Kaufman and Wigderson~\cite{KaufWig} and Kaufman and
Lubotzky~\cite{KaufLub}, are
the few that attempt to find new codes using invariance. While the
performance of their codes is very good,
unfortunately they do not
seem to lead to local testability and the performance is too good to be locally
decodable (or locally correctible). Our work seems to be the first in
this context to
explore new codes that do guarantee some locality properties.

A second, more technical, point of departure is that our work refocusses
attention on invariance of ``multivariate properties''. Since the
work of \cite{KS08-ECCC} most subsequent works focussed on
univariate properties. While this study seemed to be without loss
of generality, for the purpose of constructions it seems necessary
to go back to the multivariate setting.
One specific
contribution in this direction is that we show that invariance under
general affine-transformations and under affine-permutations lead to
the same set of properties (see Section~\ref{app:perm}).

\paragraph{Organization.}
In Section~\ref{sec:prelim} we present some of the background
material on affine-invariant codes and present some extensions
in the multivariate setting. In Section~\ref{sec:codes} we describe
our codes and analyze them. 
In Section~\ref{sec:nikodym} we describe our application to lower
bounding Nikodym sets.
In Section~\ref{sec:distance} we
describe how distance of lifted codes behave. Some of the technical
proofs are deferred to the appendix.

\paragraph{Version.} A previous version of this paper appeared,
as \cite{GS12}. The main difference in the results is the addition,
in this version, of lower bounds on the size of Nikodym sets
(Theorem~\ref{thm:nikodym}).

\section{Preliminaries}
\label{sec:prelim}

In this section we describe some basic aspects of affine-invariant
properties, specifically their degree sets.
We mention in particular
the fact that the size of degree sets determines the dimension of a given
affine-invariant code. Finally we conclude by relating the degree
set of a base code to the degree set of a lifted code. In later
sections we will use this relationship to lower bound the size of the
degree set of lifted codes, and thus lower bound their dimension.
We note that the results of this section are described for general
$q$ (and not for the special case of $q=2$).

For a function $f:\F_Q^m \to \F_q$, we associate with it the
unique polynomial in $\F_Q[x_1,\ldots,x_m]$ of degree at most
$Q-1$ in each variable that evaluates to $f$. (We abuse notation
by using the same notation to refer to a function and the associated
polynomial.) For $\vec{d} = \langle d_1,\ldots,d_m \rangle$
and $\vec x = \langle x_1,\ldots,x_m \rangle$, let $\vec x^{\vec d}$
denote the monomial $\prod_{i=1}^m x_i^{d_i}$.
For a function $f = \sum_{\vec d} c_{\vec d} \vec x^{\vec d}$, let
its support, denoted $\supp(f)$, be the set of degrees with
non-zero coefficients in $f$, i.e., $\supp(f) = \{\vec d \mid
c_{\vec d} \ne 0\}$.

\begin{definition}[Degree set]
\label{def:degree-set}
For a code $\calf \subseteq \{\F_Q^m \to \F_q\}$, its
degree set, denoted $\Deg(\calf)$, is the
set $\Deg(\calf) = \cup_{f \in \calf} \supp(f)$.
For a set $D \subseteq \{0,\ldots,Q-1\}^m$, let its
code, denoted $\Fam(\calf)$ be the set
$\Fam(\calf) = \{f:\F_Q^m \to \F_q \mid \supp(f) \subseteq D\}$.
\end{definition}

For an affine-invariant code, its degree set uniquely determines
the code and in particular the following proposition holds.

\begin{proposition}\label{prop:fam}
For linear affine-invariant codes $\calf \subseteq
\{\F_Q^m \to \F_q\}$, we have $\Fam(\Deg(\calf)) = \calf$.
\end{proposition}

We prove the proposition below. The proof uses some basic
facts about linear affine-invariant codes that are proved
in Section~\ref{app:perm}. (We note that this would be the logical
place to read/verify the contents.)

\begin{proof}
Trivially $\calf \subseteq \Fam(\Deg(\calf))$.
For the other direction,
consider $f \in \Fam(\Deg(\calf))$.
Express $f = \Tr \circ g$ (see, e.g., Proof of Lemma~\ref{lem:monomial}),
where $g \in \F_Q[\vec x]$ is chosen among all such to
be minimal in its support. We have $\supp(g) \subseteq \supp(f) \subseteq
\Deg(\calf)$. Suppose $g = \sum_{\vec d \in \Deg(\calf)}
c_{\vec d} \vec x^{\vec d}$, then by Lemma~\ref{lem:monomial}
we have $\Tr(c_{\vec d} \vec x^{\vec d}) \in \calf$ for every
$\vec d$. Now by linearity of $\calf$ it follows that
$\sum_{\vec d} \Tr(c_{\vec d} \vec x^{\vec d}) \in \calf$,
but by the linearity of the Trace function we have that this function
is $f$.
\end{proof}

Our reason to study the degree sets is that the size of
the degree set gives the dimension of a code exactly.

\begin{proposition}
For a linear affine-invariant code $\calf \subseteq
\{\F_Q^m \to \F_q\}$, we have the dimension of
$\calf$ equals $|\Deg(\calf)|$.
\end{proposition}

\begin{proof}
We generalize the proof of~\cite[Lemma~2.14]{BGMSS11-ECCC} to
the multivariate setting.
For degree $\vec d \in \{0,\ldots,Q-1\}^m$, define
$S(\vec d) = \{q^i \vec d \mid i \in \Z\}$. For every
$\vec d, \vec e$, either $S(\vec d) = S(\vec e)$ or $S(\vec d) \cap S(\vec e) = \varnothing$.
Write $f \in \calf$ as $f(\vec x) = \sum_{\vec d \in \Deg(\calf)} f_{\vec d} \vec x^{\vec d}$.
Since $f^q = f$, it follows that $f_{q \cdot \vec d} = f_{\vec d}^q$ for all $\vec d$
and hence $f_{\vec d} \in \F_{q^{|S(\vec q)|}}$.
From each $S(\vec d)$ pick a representative, and let $S$ be the set of these
representatives, so that $\Deg(\calf) = \cup_{\vec d \in S} S(\vec d)$ is a partition.
Then we may write
$f(\vec x) = \sum_{\vec d \in S} \Tr_{\F_{q^{|S(\vec d)|}}, \F_q}(f_{\vec d} \vec x^{\vec d})$.
For each $\vec d \in S$ there are $q^{|S(\vec d)|}$ choices for $f_{\vec d}$, so the total
number of choices for $f$ is
$\prod_{\vec d \in S} q^{|S(\vec d)|}
= q^{\sum_{\vec d \in S} |S(\vec d)|}
= q^{|\Deg(\calf)|}$
\end{proof}

Next we attempt to describe how the degree set of a lifted
code can be determined from the degree set of a base code.
We start by mentioning a simple property of degree sets that
will be quite useful in our analysis.

Let $\modstar{Q}$ denote the operation that maps non-negative
integers to the set $\{0,\ldots,Q-1\}$ as given by
$a \modstar Q = 0$ if $a = 0$ and $a \modstar Q = b \in \{1,\ldots,
Q-1\}$ if $a \ne 0$ and $a = b \pmod{Q-1}$. (Note that if $a \modstar
Q = b$, then $x^a = x^b \pmod{x^Q - x}$.)

For $Q = q^n$ and 
$\vec e, \vec d \in \{0,\ldots,Q-1\}^n$, we say that
$\vec e$ is a $q$-shift of $\vec d$ if there exists
$j$ such that for every $i$, we have $e_i = q^j \cdot d_i \modstar{Q}$.
Note that $\vec e$ is a $q$-shift of $\vec d$ if and only if
$\vec d$ is a $q$-shift of $\vec e$.

\begin{proposition}
\label{prop:q-shift}
Let $\calf \subseteq \{\F_Q^m \to \F_q\}$ be a linear affine-invariant
code and let $D = \Deg(\calf)$ be its degree set.
Then $D$ is $q$-shift closed, i.e., if $\vec d \in D$ and
$\vec e$ is a $q$-shift of $\vec d$ then $\vec e \in D$.
\end{proposition}

\begin{proof}
Follows immediately from the fact that for every function
$f:\F_Q^m \to \F_q$, we have $\vec d \in \supp(f)$ if
and only if $\vec e \in \supp(f)$, which follows from
the fact that $f(\vec x)^{q^j} = f(\vec x) \mod (\vec x^Q - \vec x)$
for every $j$.
\end{proof}

We now turn to identifying the degree sets of lifted codes.
We start with the case of lifts of univariate codes, which
are somewhat simpler to describe. The lifts of multivariate
codes come from the same principles, but are messier to describe.

It turns out that the structure of the degree set (not every
set $D$ is the degree set of an affine-invariant code) is strongly
influenced by the base $p$ representation of its members,
where $p$ is the characteristic of $q$, the alphabet of our codes.
We start with some notions related to such representations.
For non-negative integers $a$ and $b$, let $a^{(0)},a^{(1)},\ldots,$
and $b^{(0)},b^{(1)},\ldots,$ be their base $p$ expansion, i.e.,
$0 \leq a^{(i)},b^{(i)} < p$, $a = \sum_i a^{(i)} p^i$  and
$b = \sum_i b^{(i)} p^i$. We say $a$ is in the $p$-shadow of
$b$, denoted $a \leq_p b$, if $a^{(i)} \leq b^{(i)}$ for every $i$.
We extend the notion to vectors coordinate-wise. So for, 
$\vec e, \vec d \in \Z^n$, we say $\vec e \leq_p \vec d$ if
$e_i \leq_p d_i$ for all $i \in [n]$.

\begin{definition}
\label{def:univ-lift-deg}
For a set $D \subseteq \{0,\ldots,Q-1\}$, its $m$th lift,
denoted $\Lift_m(D)$ is given by 
\[
\Lift_m(D) \triangleq \left\{
\vec d = \langle d_1,\ldots,d_m \rangle \in \{0,\ldots,Q-1\}^m | 
\forall \vec e \leq_p \vec d, ~
\sum_{i=1}^m e_i \modstar Q \in D 
\right\}.
\]
\end{definition}

The following proposition makes the implied connection between
lifts of codes and their degree sets explicit.
We note that this proposition is
implicit in \cite{BMSS}.

\begin{proposition}
\label{prop:univ-lift-deg}
For every linear affine-invariant code $\calf \subseteq \{\F_Q \to
\F_q\}$, and for every $m \geq 1$, we have
$\Lift_m(\Deg(\calf)) = \Deg(\Lift_m(\calf))$.
\end{proposition}

\begin{proof}
Let $\F = \F_q$ and $\K = \F_Q$. In what follows
we will use the notation $\vec x^{\vec e}$ to denote
$\prod_{i=1}^n x_i^{e_i}$. And we use ${\vec d \choose
\vec e}$ to denote $\prod_{i=1}^n {d_i \choose e_i}$.

Since $\calf$ is linear, we have that there exist
some $I \leq Q$ linear constraints given by $t_{i,j}
\in \K$ and $\lambda_{ij} \in \F$ for 
$1 \leq i \leq I$ and $1 \leq j \leq J$  such that
$f \in \calf$ if and only if $\sum_{j\leq J} \lambda_{ij}
f(t_{ij}) = 0$ for every $i \leq I$.

We now have the following equivalences:
\begin{eqnarray*}
\vec d \in \Deg(\lift_m(\calf))
&\substack{\text{Lemma~\ref{lem:monomial}} \\ \iff}&
\forall \lambda \in \K~~\Tr(\lambda {\vec x}^{\vec d} ) \in \lift_m(\calf) \\
&\iff& \forall \lambda \in \K~\forall \vec a \in \K^m~
\forall \vec b \in \K^m~~\Tr(\lambda(t \cdot \vec a  + \vec b)^{\vec d}) \in \calf \\
&\iff& \forall \lambda, \vec a, \vec b~~\forall i~~
\sum_j \lambda_{ij} \Tr(\lambda(t_{ij} \vec a + \vec b)^{\vec d}) = 0 \\
&\substack{\text{Lemmas~\ref{lemma:genlucas},~\ref{lemma:expansion}}  \\ \iff}&
\forall \lambda, \vec a, \vec b~~\forall i~~
\Tr\left(\lambda
\sum_{\vec e \le_p \vec d} {\vec d \choose \vec e}
{\vec a}^{\vec e} {\vec b}^{\vec d - \vec e}
\sum_j \lambda_{ij} {t_{ij}}^{\sum_{\ell=1}^n e_\ell}
\right) =0 \\
&\iff& \forall \vec a, \vec b~~\forall i~~
\sum_{\vec e \le_p \vec d} {\vec d \choose \vec e}
{\vec a}^{\vec e} {\vec b}^{\vec d - \vec{e})}
\sum_j \lambda_{ij} {t_{ij}}^{\sum_\ell e_\ell} = 0 \\
&\iff& \forall \vec e \le_p \vec d~~\forall i~~
\sum_j \lambda_{ij} {t_{ij}}^{\sum_\ell e_\ell} = 0 \\
&\iff& \forall \vec e \le_p \vec d~~\Sigma(\vec e) \modstar{Q} \in
\Deg(\calf).
\end{eqnarray*}
\end{proof}

We now extend the above definition and proposition to the 
case where the code being lifted is itself a multivariate one.

To this end we extend some of the notations from the previous parts to
matrices.
For matrices $\bA, \bB \in \Z^{n \times \ell}$ we say 
$\bA \leq_p \bB$ if $(\bA)_{ij} \leq_p (\bB)_{ij}$ for every
pair $(i,j) \in [n] \times [\ell]$.

%\mnote{I changed the definition below. Seems we don't need $S$,
%since it is already implied by the forall quantifier over $\vec f$.}
\iffalse
Finally, we extend the notion 
to compare matrices to vectors.
For
$\vec e \in \Z^\ell$ and $d \in \Z$
we say $\vec e \leq_p d$ if for every subset $S \subseteq [\ell]$
and for every $\vec f \leq_p \vec e$
we have $\sum_{i \in S} f_i \leq_p d$. (This notion corresponds to
the support of $(1 + \sum_{i=1}^\ell x_i)^d$: $\vec x^{\vec e}$ appears
with a non-zero coefficient only if $\vec e \leq_p d$.)
\fi
Next, we extend the notion 
to compare vectors to elements and matrices to vectors.
For
$\vec e \in \Z^\ell$ and $d \in \Z$
we say $\vec e \leq_p d$ if 
for every $\vec f \leq_p \vec e$
we have $\sum_{i \in [\ell]} f_i \leq_p d$. (This notion corresponds to
the support of $(1 + \sum_{i=1}^\ell x_i)^d$: $\vec x^{\vec e}$ appears
with a non-zero coefficient only if $\vec e \leq_p d$.)
Extending to matrices and vectors,
$\bA \in \Z^{n \times \ell}$ with rows $(\bA)_j \in \Z^\ell$
and
$\vec d = \langle d_1,\ldots,d_n \rangle \in \Z^n$
we say $\bA \leq_p \vec d$ if
$(\bA)_j \leq_p d_j$ for every $j \in [n]$.

Finally, we need one more piece of notation before defining the
degree sets of multivariate lifts. 
For matrix $\bA \in \Z^{n \times \ell}$,
let $\Sigma(\bA) \in \Z^\ell$ denote its row sum given by
$\Sigma(\bA)_j = \sum_{i=1}^n (\bA)_{ij}$.

We are now ready to define the lifts of multivariate degree sets.

\begin{definition}[Degree sets of lifts]
For a set $D \subseteq \{0,\ldots,Q-1\}^t$, its
$m$th {\em lift}, denoted $\Lift_m(D)$, is given by
$$\{\vec d \in \{0,\ldots,Q-1\}^m \mid \forall~
\vec E \in \Z^{m \times t} \leq_p \vec d, \mbox{ we have }
\Sigma(\vec E) \modstar Q \in D\}.$$
\end{definition}

The following proposition is the multivariate analog of 
Proposition~\ref{prop:univ-lift-deg}.

\begin{proposition}
\label{prop:lift-deg}
For every linear affine-invariant code $\calf \subseteq \{\F_Q^t \to
\F_q\}$, and for every $m \geq t$, we have
$\Lift_m(\Deg(\calf)) = \Deg(\Lift_m(\calf))$.
\end{proposition}

\begin{proof}
The proof is very similar to the proof of
Proposition~\ref{prop:univ-lift-deg}, with enriched notation.
For a matrix $\vec E$, let $\vec E^{\top}$ denote the transpose of $\vec E$,
so that $(\vec E)_{ij} = (\vec E^\top)_{ji}$.
For an integer $d$ and vector $\vec e = \langle e_1,\ldots,e_t \rangle$, define
${d \choose \vec e} = \frac{d!}{e_1!\cdots e_t!(d-e_1-\cdots-e_n)!}$, the
standard multinomial coefficient and also the coefficient of $\vec x^{\vec e}$
in the expansion of $(1 + \sum_{i=1}^t x_i)^d$.
Extending this notation, for a vector $\vec d = \langle d_1,\ldots,d_m \rangle$
and a matrix $\vec E \in \Z^{m \times t}$ with rows $\vec e_1,\ldots,\vec e_m$,
define ${\vec d \choose \vec E} = \prod_{i=1}^m {d_i \choose \vec e_i}$.
We will use the fact that ${\vec d \choose \vec E} \not\equiv 0 \pmod{p}$ if and only if
$\vec E \le_p \vec d$ (see Lemma~\ref{lemma:genlucas}).
Finally, for two matrices $\vec A, \vec E \in \Z^{m \times t}$, define
$\vec A^{\vec E} = \prod_{i,j} a_{ij}^{e_{ij}}$.
For convenience, let $\F = \F_q$ and let $\K = \F_Q$.

Now, we begin the proof.
There exist $\vec t_{ij} \in \K^t,\lambda_{ij} \in \F$ such that $f \in \calf \iff \forall i~
\sum_j \lambda_{ij} f(\vec t_{ij}) =~0$.
The assertion then follows from the
following equivalences:
\begin{eqnarray*}
\vec d \in \Deg(\lift_m(\calf))
&\substack{\text{Lemma~\ref{lem:monomial}} \\ \iff}&
\forall \lambda \in \K~~\Tr(\lambda {\vec x}^{\vec d} ) \in \lift_m(\calf) \\
&\iff& \forall \lambda \in \K~\forall \vec A \in \K^{m \times t}~
\forall \vec b \in \K^m~~\Tr(\lambda(\vec A \vec x + \vec b)^{\vec d}) \in \calf \\
&\iff& \forall \lambda, \vec A, \vec b~~\forall i~~
\sum_j \lambda_{ij} \Tr(\lambda(\vec A \vec t_{ij} + \vec b)^{\vec d}) = 0 \\
&\substack{\text{Lemmas~\ref{lemma:genlucas},~\ref{lemma:expansion}}  \\ \iff}&
\forall \lambda, \vec A, \vec b~~\forall i~~
\Tr\left(\lambda
\sum_{\vec E \le_p \vec d} {\vec d \choose \vec E}
{\vec A}^{\vec E} {\vec b}^{\vec d - \Sigma(\vec E^{\top})}
\sum_j \lambda_{ij} {\vec t_{ij}}^{\Sigma (\vec E)}
\right) \\
&\iff& \forall \vec A, \vec b~~\forall i~~
\sum_{\vec E \le_p \vec d} {\vec d \choose \vec E}
{\vec A}^{\vec E} {\vec b}^{\vec d - \Sigma(\vec E^{\top})}
\sum_j \lambda_{ij} {\vec t_{ij}}^{\Sigma (\vec E)} = 0 \\
&\iff& \forall \vec E \le_p \vec d~~\forall i~~
\sum_j \lambda_{ij} {\vec t_{ij}}^{\Sigma (\vec E)} = 0 \\
&\iff& \forall \vec E \le_p \vec d~~\Sigma(\vec E) \modstar{Q} \in \Deg(\calf)
\end{eqnarray*}
\end{proof}

%We note that the special
%case where $t=1$ was shown implicitly noted in \cite{BMSS}.
%Here we generalize the proposition to general $t$.
%\toedit{
%We skip the proof from this version.}
%The proof follows
%easily from the definition
%using the fact that if we consider some affine
%transformation $A:\F_Q^m \to \F_Q^m$ and some
%monomial $\vec x^{\vec d}$ then the function $(A(\vec x))^{\vec d}$
%is supported on monomials of the form $\Sigma(\vec E)$
%for $\vec E \leq_p \vec d$.

The definition of $\Lift_m(D)$ is somewhat cumbersome and
not easy to work with. However in the upcoming sections we
will try to gain some combinatorial insights about it
to derive bounds on the dimension of the codes of interest.

Finally, before concluding we mention explicitly the locality
properties of lifted codes. We start with a simple observation.

\begin{proposition}
\label{prop:distance-simple}
Let $\calf \subsetneq \{\F_Q^m \to \F_q\}$ be a linear
affine-invariant code. The $\delta(\calf) \geq 2\cdot Q^{-m}$.
\end{proposition}

\begin{proof}
For $\vec a \in \F_Q^m$ let 
$\Delta_{\vec a}:\F_Q^m \to \F_q$ be the function satisfying
$\Delta_{\vec a}(\vec a) =1$ and $\Delta_{\vec a}$ is zero
everywhere else. For contradiction assume $\Delta_{\vec a} \in \calf$
for some $\vec a \in \F_Q^m$. But then by affine-invariance
we have $\Delta_{\vec b} \in \calf$ for every $\vec b \in \F_Q^m$
and then by linearity we have every function in
$\{\F_Q^m \to \F_q\}$ is contained in $\calf$, contradicting 
our hypothesis on $\calf$.
\end{proof}

\begin{proposition}
\label{prop:liftlcc}
Let $\calf \subsetneq \{\F_Q^t \to \F_q\}$ be a linear
affine-invariant code. Let $\calL = \Lift_m(\calf)$ be
its $m$-ary lift. Then $\calL$ is a $(Q^t-1,\frac13Q^{-t})$-LCC
and a $(Q^t,\Omega(Q^{-2t}),Q^{-t})$-LTC.
\end{proposition}

\begin{proof}
Given $f : \F_Q^m \to \F_q$ that is $Q^{-t}/3$-close to $p \in \calL$
and $\vec a \in \F_Q^m$, the local decoding algorithm
works as follows:
Pick random linearly independent $\vec b_1,\ldots \vec b_t \in \F_Q^m$ and let
$h : \F_Q^t \to \F_q$ be given by 
$h(\vec 0) = 0$ and
$h(u_1,\ldots,u_t) = f(\vec a + u_1 \vec b_1 + \cdots
+ u_t \vec b_t)$ for all other $u_1,\ldots,u_t$.
Compute $g \in \calf$ such that $g(\vec u) = h(\vec u)$ for all
$\vec u \in \F_Q^t \setminus \{\vec 0\}$ and
output~$g(\vec 0)$.

It is clear that the decoder makes at most $Q^t - 1$ queries.
We show that the decoder succeeds with high probability.
Let $p \in \calf$ satisfy $\delta(p,f) \leq Q^{-t}/3$.
Let $A$ be the $t$-dimensional subspace
$A = \{\vec a + u_1 \vec b_1 + \cdots + u_t \vec b_t \mid \vec u \in \F_Q^t\}$.
For every $\vec u \in \F_Q^t \setminus \{\vec 0\}$ we have
$\Pr_{\vec b_1,\ldots,\vec b_t}
[h(\vec u) \ne p|_A(\vec u)] \leq Q^{-t}/3$.
By a union bound 
$\Pr_{\vec b_1,\ldots,\vec b_t}[
\exists \vec u \in \F_Q^t \setminus \{\vec 0\} | 
h(\vec u) \ne p|_A(\vec u)] \leq (Q^t - 1)/(3Q^t) < 1/3$. 
So, with probability at least $2/3$, we have that $p|_A \in \calf$ 
agrees
with $h$ on all of $\F_Q^t - \vec{0}$. Furthermore,
by the fact that $\delta(\calf) \geq 2Q^{-t}$
(Proposition~\ref{prop:distance-simple}), $p|_A$ is the unique such
function with this property. It follows that the decoder outputs
$p|_A(\vec 0) = p(\vec a)$ with probability at least $2/3$ as
desired. 

The local testability follows directly from~\cite[Theorem~2.9]{KS08-ECCC}.
\end{proof}

\section{Constructions}
\label{sec:codes}

\subsection{Codes of constant locality}
\label{ssec:one}

In this section we prove Theorem~\ref{thm:one} which promised binary
codes
of locality $k$ and length $N$ with dimension $\Omega_k(\log N)^{k-2}$.

\paragraph{The Code:}
Fix $k = Q = 2^\ell$ and $N = 2^{m\ell}$. Let
$\calf_1 \subseteq \{\F_{Q} \to \F_2\}$ be the
code given by $\{f:\F_Q \to \F_2 \mid \sum_{\alpha \in \F_Q}
f(\alpha) = 0\}$.
Let $\calL_1 = \Lift_m(\calf_1)$. In what follows
we verify that $\calL_1$ has the properties claimed
in Theorem~\ref{thm:one}.

We start with some obvious aspects.

\begin{proposition}
\label{prop:length-one}
$\calL_1$ is a binary code of length $N$ and a $(k-1,k^{-1}/3)$-LCC
and a $(k,\Omega(k^{-2}),k^{-1})$-LTC.
\end{proposition}

\begin{proof}
The length is immediate from the construction.
The local correctability and testability follow from
Proposition~\ref{prop:liftlcc}.
\end{proof}

The main aspect to be verified is the dimension of $\calL_1$.
We first describe the degree set of $\calf_1$.

\begin{claim}\label{clm:deg-one}
$\Deg(\calf_1) = \{0,\ldots,Q-2\}$.
\end{claim}

\begin{proof}
Write $f : \F_Q \to \F_2$ as $f(x) = \sum_{d=0}^{Q-1} f_d x^d$.
Then
$\sum_{\alpha \in \F_Q} f(\alpha)
= \sum_{\alpha \in \F_Q} \sum_{d=0}^{Q-1} f_d \alpha^d
= \sum_{d=0}^{Q-1} f_d \left( \sum_{\alpha \in \F_Q} \alpha^d \right)
= -f_{Q-1}$ where we have used the fact that
$\sum_{\alpha \in \F_Q} \alpha^d = -1$ if $d=Q-1$ and is equal to $1$ otherwise.
Therefore $f \in \calf_1$ if and only if $\deg(f) < Q-1$.
\end{proof}

{\bf Remark:} Note that the proof above applies without change
to the case of the range being $\F_q$, for any $q$,
provided $\F_Q$ extends $\F_q$.

The next claim interprets the definition of $\Lift_m(D)$
in our setting.

\begin{claim}
$\vec d \in \{0,\ldots,Q-1\}^m$ is contained in $\Lift_m(\Deg(\calf_1))$
if and only if for every $\vec e \leq_2 \vec d$ we have
$\sum_{i=1}^m e_i \modstar Q \ne Q-1$.
\end{claim}

\begin{proof}
Follows immediately by applying Proposition~\ref{prop:univ-lift-deg}
to Claim~\ref{clm:deg-one}.
\end{proof}

Given the claim, it is simple to get a lower bound on the
dimension of our code.

\begin{lemma}
\label{lem:dim-one}
The dimension of $\calL_1$ is at least ${m \choose Q-2}$.
\end{lemma}

\begin{proof}
For $S \subseteq [m]$ let $\vec {d_S}$ denote the vector
that is one on coordinates from $S$ and zero outside.
It is clear that for $|S| \leq Q-2$, $\vec{d_S} \in \Lift_m(\Deg(\calf_1))$
and there are at least ${m \choose Q-2}$ such sets.
\end{proof}

\begin{proof}[Proof of Theorem~\ref{thm:one}]
Theorem~\ref{thm:one} follows Proposition~\ref{prop:length-one}
and Lemma~\ref{lem:dim-one} and plugging the values of $m$
and $Q$ from the construction. Specificalle we have that
the dimension of the code is at least
${m \choose
Q-2} \geq \frac{1}{k^{k-2} k!} (\log N)^{k-2}$.
So the theorem follows for $c_k = \frac{1}{k^{k-2} k!}$.
\end{proof}

\subsection{Codes of sublinear locality}
\label{ssec:two}

Next we turn to Theorem~\ref{thm:two}, which asserts the existence
of codes of locality $N^{\epsilon}$ with dimension
$N - N^{1 - \epsilon'}$.

\paragraph{The Code:}
Given $\epsilon > 0$ and prime $p$, 
let $m = \lceil 1/\epsilon \rceil$.
Let $\ell$ be an integer such that $p^{m\ell} \geq N$.
Let $Q = p^\ell$.
Let $\calf_2 \subseteq \{\F_Q \to \F_p\}$ be
the code $\{f:\F_Q \to \F_p \mid \sum_{\alpha \in \F_Q} f(\alpha) = 0 \}$.
Let $\calL_2 = \Lift_m(\calf_2)$.

As usual we get the following proposition.

\begin{proposition}
\label{prop:length-two}
$\calf_2$ is a $p$-ary code of length at least $N$ and
locality at most $N^{\epsilon}$. Specifically
it is a $(N^{\epsilon},\Omega(N^{-\epsilon}))$-LCC
and a $(N^{\epsilon},\Omega(N^{-2\epsilon}),N^{-\epsilon})$-LTC.
\end{proposition}

We now turn to the task of analyzing the dimension of this code.
We first describe the degree sets of $\calf_2$ and $\calL_2$.

\begin{claim}
\label{clm:deg-two}
$\Deg(\calf_2) = \{0,\ldots,Q-2\}$ and
$$\Deg(\calL_2) = \{\vec d \in \{0,\ldots,Q-1\}^m \mid
\forall \vec e \leq_p \vec d, \sum_i e_i \modstar Q \ne Q-1\}.$$
\end{claim}

\begin{proof}
The first part follows from the proof of Claim~\ref{clm:deg-one}
(see the remark following the proof).
The second part follows immediately from Proposition~\ref{prop:lift-deg}.
\end{proof}

\begin{lemma}
\label{lem:dim-two}
The dimension of $\calL_2$ is at least $N - N^{1 - \epsilon'}$
for some $\epsilon' = \Omega(2^{-2/\epsilon})$.
\end{lemma}

\begin{proof}
Let $D = \Deg(\calf_2)$.
Let $\vec e = \langle e_1,\ldots,e_m \rangle$ and
$e_i^{(0)},e_i^{(1)},\ldots,e_i^{(\ell-1)}$ denote the
$p$-ary expansion of $e_i$.

\begin{claim}
\label{clm:lift-deg-two}
If there exists integer $s \in \{0,\ldots,\ell-1\}$
such that for every $i\in [m]$ and every $j \in [1 + \lceil \log m \rceil ]$ we have
$e_i^{(s + j \pmod \ell)} = 0$, then
$\vec e \in \Lift_m(D)$.
\end{claim}

\begin{proof}
Recall, by Proposition~\ref{prop:q-shift} that
$\vec e \in \Lift_m(D)$ if and only if $\vec{e'} \in D$
for every $\vec{e'}$ that is a $p$-shift of $\vec e$.
Thus without loss of generality we can
assume (by shifting $\vec e$ appropriately), that the block of
zeroes are the most significant digits in the $e_i$'s. (i.e., $s = \ell -
\lceil \log m \rceil - 2$.)

With this assumption, we now have $e_i < p^{\ell - \log m - 1}
< Q/(pm) < (Q-1)/m$. We thus conclude that for
every $\vec f \leq_2 \vec e$, $\sum_{i=1}^m f_i
\leq \sum_{i=1}^m e_i < Q-1$ and so (by Claim~\ref{clm:deg-two})
$\vec e \in \Lift_m(D)$.
\end{proof}

The lemma follows by an easy counting argument. Let $t = 1 + \lceil \log m
\rceil$.
We partition the
set $[\ell]$ into $\ell/t$ blocks of $t$ successive integers each.
For each such block the number of possible assignments of digits that
do not make the entire block zero in each $e_i$ is
$p^{mt} - 1$. Thus the total number of vectors $\vec e$ that do not
have any of these blocks set to zero is
$(p^{mt} - 1)^{\ell/t} = p^{m\ell}(1 - p^{-mt})^{\ell/t}
\approx p^{m \ell} e^{\ell/(tp^{mt})} = p^{m \ell (1 - \Omega(1/mt p^{mt}))}
= N^{1 - {\epsilon'}}$ for
$\epsilon' = 1/({mt} p^{mt})$. Recalling that $\epsilon = 1/m$,
we have $\epsilon' = \Omega(p^{-2/\epsilon})$.
The lemma follows by noting that if $\vec e \not\in \Lift_m(D)$
then in each of these blocks it must be non-zero somewhere
(by Claim~\ref{clm:lift-deg-two} above).
\end{proof}

\begin{proof}[Proof of Theorem~\ref{thm:two}]
Theorem~\ref{thm:two} follows immediately from
Proposition~\ref{prop:length-two} and Lemma~\ref{lem:dim-two}.
\end{proof}

%\mnote{Changed the section to work for general $p$ so following
%(now commented out) remark is not needed.}
\iffalse
\paragraph{Construction over characteristic $p$:}
The above construction works for fields over any constant characteristic
$p$. More precisely, fix $\epsilon$ and $m$ as before, and
let $\ell$ be an integer such that $p^{m\ell} \ge N$ and let
$Q = p^\ell$. Let $\calf_2 = \{f : \F_Q \to \F_p \mid
\sum_{\alpha \in \F_Q} f(\alpha) = 0\}$ and let
$\calL_2 = \lift_m(\calf_2)$.
Then Claim~\ref{clm:deg-two} still holds.
Moreover, the proof of Lemma~\ref{lem:dim-two} goes through
if we replace $t = 1 + \lceil \log m \rceil$ with
$t = 1 + \lceil \log_p m \rceil$ and hence $\epsilon'
= 1/(mtp^{mt}) = \Omega(p^{-2/\epsilon})$.
\fi

\subsection{Codes of linear locality}
\label{ssec:three}

Finally, we prove Theorem~\ref{thm:three}, which claims codes of locality
$\epsilon N$ with dimension $N - \poly\log N$.
This construction is different from the previous two in that here we
lift a multivariate code, whereas in both previous constructions we lifted
univariate codes.

\paragraph{The Code:}
Let $\ell = \lceil \log 1/\epsilon \rceil$ (so that $2^{-\ell}
\leq \epsilon$). Let $Q = 2^{\ell}$.
For integer $m$ let $N = 2^{m \ell}$
and let $t = m-1$.
Let $\calf_3 \subseteq \{\F_Q^t \to \F_2\}$
be given by
$\calf_3 = \{f:\F_Q^t \to \F_2 \mid \sum_{\vec \alpha \in \F_Q^t}
f(\vec \alpha) = 0 \}$.
Let $\calL_3 = \Lift_m(\calf_3)$.

\begin{proposition}
$\calL_3$ is a code of block length $N$ with locality
$\epsilon N$. Specifically, it is a $(\epsilon N, \frac13(\epsilon N)^{-1})$-LCC
and a $(\epsilon N, (\epsilon N){-2},(\epsilon N)^{-1})$-LTC.
\end{proposition}

The proposition below asserts that
every degree except the vector that is $Q-1$ in every
coordinate is in the degree set of $\calf_3$.
(Here $(Q-1)^t$ denotes the $t$-tuple all of whose entries is $Q-1$, rather than
$(Q-1)$ exponentiated to the $t$-th power).

\begin{proposition}
$\Deg(\calf_3) = \{0,\ldots,Q-1\}^t - \{(Q-1)^t\}$.
\end{proposition}

\begin{proof}
Write $f : \F_Q^t \to \F_2$ as
$f(\vec x) = \sum_{\vec d \in \{0,\ldots,Q-1\}^t} f_{\vec d} \vec x^{\vec d}$.
Then
$\sum_{\alpha \in \F_Q^t} f(\alpha)
= \sum_{\alpha \in \F_Q^t} \sum_{\vec d} f_{\vec d} \vec x^{\vec d}
= \sum_{\vec d} f_{\vec d} \left( \sum_{\alpha \in \F_Q^t} \alpha^{\vec d} \right)
= \sum_{\vec d} f_{\vec d} \prod_{i=1}^t \left( \sum_{\alpha \in \F_Q} \alpha^{d_i} \right)
= (-1)^t f_{(Q-1)^t}
$
where we have used the fact that $\sum_{\alpha \in \F_Q} \alpha^d = -1$ if
$d = Q-1$ and is equal to $0$ otherwise. Therefore $f \in \calf_3$ if and only if
$f_{(Q-1)^t} = 0$.
\end{proof}

While in general degree sets of lifts of multivariate families
are not easy to characterize, in this particular case we have
a clean characterization of the degree set.

Given $\vec e = \langle e_1,\ldots,e_m \rangle$ let
$e_i^{(j)}$ denote the $j$th bit in the binary expansion
of $e_i$. Let $M(\vec e)$ denote the $m \times \ell$
matrix with entries $M(\vec e)_{i,j} = e_i^{(j)}$.

\begin{lemma}
$\vec e \in \Lift_m(\Deg(\calf_3))$ if and only if
there exists a column in $M(\vec e)$ with at least
two zeroes.
\end{lemma}

\begin{proof}
As in the proof of Lemma~\ref{lem:dim-two} we have that
$\vec e \in \Lift_m(\Deg(\calf_3))$ if and only if
$2\vec e \modstar Q \in \Lift_m(\Deg(\calf_3))$.
So without loss of generality we can assume that
$\vec e$ is shifted so that the two zeroes are in
the most significant bits. Thus we have that
$m-2$ of the $e_i$'s, say $e_1,\ldots,e_{m-2}$,
are at most $Q-1$ and the remaining two are at most
$Q/2 - 1$. We thus have that $\sum_{i=1}^m e_i
< (m-1)Q - 1$. Using this and applying Proposition~\ref{prop:lift-deg}
it is easy to verify that $\vec e$ is not in $\Lift_m(\Deg(\calf_3))$.
\end{proof}

The following lemma now follows by simple counting.

\begin{lemma}
The dimension of $\calL_3$ is $2^{m \ell} - (m+1)^{\ell}$.
\end{lemma}

\begin{proof}[Proof of Theorem~\ref{thm:three}]
Follows by plugging in the
values for the parameters, specifically by setting
$\ell = \log 1/\epsilon$ and $m = (\log N/\log 1/\epsilon)$.
We get that the dimension of $\calL_3$ is $N - (1 + \log N/\log 1/\epsilon)^{
\log 1/\epsilon}$.
\end{proof}

We remark that the construction in \cite{BGHMRS} is very close in
parameters. In their construction (i.e., the Reed-Muller codes)
the matrix $M(\vec e)$ must have at least $\ell+1$
zeroes. Since any such matrix must have two zeroes in a single column
it follows that every matrix their construction admits is also admissible in
ours, while our allow for other matrices also. However the difference
between the length and dimension is at most a constant factor
(depending on $\ell$). (More precisely, the dimension of their code
is $2^{m \ell} - \sum_{i=0}^\ell {m \ell \choose i} \approx 2^{m\ell}
- (e m)^{\ell}$.) Of course, for their application the code needs to
have much better local testability than given here. But the local
testability given
here is just what follows immediately from the definition and previous
works, and it is quite possible that better bounds can be achieved by
more careful examination of this code.

\subsection{High-rate high-error LCCs}

Finally, we prove Theorem~\ref{thm:four}.
This construction is a departure from the others in that the
code is not binary, and the code being lifted is not the parity
check code. Finally the decoding algorithm is a bit more complex
to explain, though even this algorithm is by now folklore.

The code itself is a generalization of the classical multivariate polynomial code.
Here we consider the set of all functions $f: \F_q^m \to \F_q$ such that
the restriction of $f$ to any line has degree $d$. As is well known,
every multivariate polynomial of degree at most $d$ is such a function.
The remarkable fact is that if $q$ has small characteristic, then there are
many more such functions.

\paragraph{The Code:}
Recall that we are given $\delta$, $\epsilon$ and some $N_0$ and
we wish a code of length $N \geq N_0$ of dimension $(1-\epsilon)N$
and locality $N^{\delta}$.
Let $m = \lceil 1/\delta \rceil$ and $s$ be such that
$Q = 2^s \geq N_0^{\delta}$.
Let $b = 1 + \lceil \log m \rceil$ and
$c = \lceil b 2^{bm} \log 1/\epsilon \rceil$.
Let $\gamma = 2^{-c}$ and $\tau = \gamma/6$ (so
that
$6\tau \le \gamma \le \epsilon^{-(1 + \lceil \log m \rceil)2^{m(1 + \lceil \log m
\rceil)}}$
and let $d = (1 - 2^{-c})Q$.
Let $\calf_4 =
\{f : \F_Q\to \F_Q \mid \deg(f) \leq d\}$.
Let $\calL_4 = \Lift_m(\calf_4)$.
In words, it is the set of all degree $m$-variate functions that have degree at most 
$d$ when restricted to a line.

\paragraph{Decoding:}

The general idea for decoding $\calL_4$ is the same as that for
multivariate polynomials, and in particular the algorithm from
Gemmell et al.~\cite{GLRSW}.

Given $f:\F_Q^m \to \F_Q$ that is $\tau$-close to $p \in \calL_4$
and $\vec a \in \F_Q^m$, the decoding algorithm works as follows:
Pick a random $\vec b \in \F_Q^m$ and let $h: \F_Q \to \F_Q$
be given by $h(t) = f(\vec a + t \vec b)$. Compute, using a
Reed-Solomon decoder (see, for instance,~\cite[Appendix]{GemmellS}),
a polynomial
$g \in \F_Q[t]$ of degree at most $d$ such that $\delta(h,g) <
\gamma/2$. Output $g(0)$.

\begin{lemma}
\label{lemma:lcc-four}
$\calL_4$ is a code of block length $N$ with locality
$N^{\delta}$. Specifically, it is a $(N^{\delta}, \gamma/6)$-LCC.
\end{lemma}

\begin{proof}
Let $L = \{\vec a + t \vec b \mid t \in \F_Q - \{0\}\}$ be the line through
$\vec a$ with slope $\vec b$.
We first claim that with probability at least $2/3$, the line
$L$ contains fewer that $\gamma/2$ fraction errors (i.e.,
points $t \ne 0$ such $h(t) \ne p|_L(t)$).

\begin{claim}
For every $\vec a$, $\Pr_{\vec b} [ \delta(h,p|_L)  \geq \gamma/2 ]
< 2\tau/\gamma$.
\end{claim}

The above claim follows easily from an application of Markov's
inequality.
Next we note that if the fraction of errors on $L$ is less than
$\gamma/2$ then the decoder satisfies $g= p|_L$ and so outputs
$g(0) = p|_L(0) = p(\vec a)$ as desired.
\end{proof}

Next we turn to the analysis of the dimension of $\calL_4$ which
is similar to the analysis of $\calL_2$.
First we note the obvious fact.

\begin{proposition}
$\Deg(\calf_4) = \{0,\ldots,d\}$ and
$$\Deg(\calL_4) = \left\{\vec d \in \{0,\ldots,Q-1\}^m ~|~ \forall~\vec e \leq_2
\vec d, \sum_{i=1}^m e_i \modstar Q \in \{0,\ldots,d\}\right\}.$$
\end{proposition}

\begin{lemma}
\label{lemma:dim-four}
The dimension of $\calL_4$ is at least $(1 - \epsilon)N$.
\end{lemma}

\begin{proof}
For non-negative integer $b$, let $b^{(j)}$ denote its binary expansion
so that $b = \sum_j b^{(j)} 2^j$.
Recall $d = (1 - 2^{-c})Q$. Letting $d^{(j)}$ denoting its
binary expansion, we note an integer $e \in \{0,\ldots,Q-1\}$
is at most $d$ if (and only if) one of the bits $e^{(s-c)},\ldots,e^{(s-1)}$
is zero. We use this to reason about $\Deg(\calL_4)$.

Let $\vec d = \langle d_1,\ldots,d_m \rangle$ and let $d_i^{(j)}$
denote the $j$th bit in the binary expansion of $d_i$.

\begin{claim}
Let $b = 1 + \lceil \log m \rceil$.
If there exists $j \in \{s - c,\ldots,s-b\}$ such that
for every $i \in [m]$ and
every $\ell \in \{0,\ldots,b-1\}$
we have
$d_i^{(j+\ell)} = 0$,
then
$\vec d \in \Deg(\calL_4)$.
\end{claim}

\begin{proof}
Let $\vec e = \langle e_1,\ldots,e_m \rangle \le_2 \vec d$ and let
$e = \sum_{i=1}^m e_i \modstar{Q}$. We claim that $e^{(j+b-1)} = 0$,
which suffices to show that $e \le d$.
Let $\overline e_i = 2^{s-(j+b)}e_i \modstar{Q}$ for all $i \in [m]$,
and let $\overline e = \sum_{i=1}^m \overline e_i \modstar{Q}$.
For every $i \in [m]$ and every $k \in [s]$,
$e_i^{(k)} = \overline e_i^{(k+s-(j+b) \pmod{s})}$ and similarly
$e^{(k)} = \overline e^{(k+s-(j+b) \pmod{s})}$. Therefore it suffices to show
that $\overline e^{(s-1)} = 0$ or equivalently $\overline e < 2^{s-1}$.
By our assumption on $\vec d$,
$e_i^{(j+\ell)} = 0$ for all $\ell \in \{0,\ldots,b-1\}$, so
$\overline e_i^{(k)} = 0$ for all $k \in \{s-b,\ldots,s-1\}$ and
thus $\overline e_i < 2^{s-b}$ for all $i \in [m]$. By our choice of $b$,
$m \le 2^{b-1}$, and thus
$\sum_{i=1}^m \overline e_i < m2^{s-b} \le 2^{s-1}$.
\end{proof}

We now consider picking $\vec d$ at random.
By partitioning the $c$ most significant bits into
disjoint blocks of $b$ bits each, we get
that any such block is all zero with
probability at least $2^{-mb}$. Thus the probability there exists
a block which is all zero is at least $1 - (1 - 2^{-mb})^{c/b}
\geq 1 - e^{-c/(b2^{mb})}$.
By choice of $c$ we have that $c/(b2^{mb}) \geq \ln (1/\epsilon)$
and so $e^{-c/(b 2^{mb})} \leq \epsilon$ and thus the dimension
is lower bounded by $(1 - \epsilon)N$.
\end{proof}

\begin{proof}[Proof of Theorem~\ref{thm:four}]
%Theorem~\ref{thm:four} now follows immediately.
Follows immediately from Lemmas~\ref{lemma:lcc-four}
and~\ref{lemma:dim-four}.
\end{proof}

We remark that the construction of this section is somewhat
contrary to folk belief, which tends to suggest that
generalized Reed-Muller
codes (evaluations of $m$-variate polynomials of degree at most $d$)
are equivalently defined by requiring that their restriction
to lines are Reed-Solomon codewords (evaluations of univariate
degree $d$ polynomials). As pointed out earlier this folk statement
is true only with some restrictions on $d$ and $Q$, and our construction
benefits by violating the restrictions. While the fact that there exist
functions that are not degree $d$ polynomials, for $d \geq Q - Q/p$,
which are degree $d$ polynomials on every line has been known for
a while~\cite{FriedlS}, presumably it was suspected that the
effect on the dimension of the lifted family was negligible. Fortunately for this work,
this presumption turned out to be false.

We also give below an example of some concrete setting of parameters
for which this construction works.

\begin{example}
For every $N = 2^{2n}$, for $n \geq 7$, there exists a code of 
length $N$ over the alphabet $\F_{2^n}$
of dimension $.77N$ that is 
decodable from $0.26\%$ fraction errors with $\sqrt{N}$ queries
\end{example}

The example is obtained by setting $c = 6$, $m=2$ and $Q = 2^n$
in the construction. The fraction of errors is $2^{-6}/6 \approx 0.26\%$.
The rate follows from the following claim.

\begin{claim}
The dimension of the code is at least $((4^c - (5/4) 3^c + 1/4)/4^c) N$.
\end{claim}

While the error-correction rate of the code is smaller than that in
\cite{KSY}, it does seem to start working at much smaller lengths
and with much smaller alphabet sizes.

\section{Nikdoym sets}
\label{sec:nikodym}

A \emph{Nikodym set} $N \subseteq \F_q^m$ is a set such that for all
$x \in \F_q^m$, there exists $y \in \F_q^m$ such that the punctured line
$\{x + ty \mid t \in \F_q \setminus \{0\}\} \subseteq N$.

%We begin by recalling a simple fact about the minimum distance of
%nontrivial linear affine-invariant families.
%\begin{prop}
%\label{prop:nontrivial_dist}
%If $\calf \subsetneq \{\F_q^m \to \F_q\}$, then $\wt(f) \ge 2$ for all
%nonzero $f \in \calf$.
%\end{prop}
%This is a special case of Proposition~\ref{prop:distance-simple}.

%\begin{proof}
%Suppose $f \in \calf$ and $\wt(f) = 1$. There is some $\alpha \in \F_q^m$ such that
%$f(\alpha) \ne 0$ and $f(\beta) = 0$ for all $\beta \ne \alpha$.
%By linearity, we may assume that
%$f(\alpha) = 1$. For each $\gamma \in \F_q^m$, define $f_\gamma(x) = f(x+\alpha-\gamma)$.
%By affine-invariance, $f_\gamma \in \calf$ for all $\calf$. But for each
%$\gamma \in \F_q^m$, we have $f(\gamma) = 1$ and $f(\gamma') = 0$ for all
%$\gamma' \ne \gamma$, so the set $\{f_\gamma\}_{\gamma \in \F_q^m}$ spans
%$\{\F_q^m \to \F_q\}$, since any arbitrary function $g : \F_q^m \to \F_q$ may
%be written as $g = \sum_{\gamma \in \F_q^m} g(\gamma)f_\gamma$. This contradicts
%the fact that $\calf$ is nontrivial.
%\end{proof}

The following proposition strengthens and generalizes the result usually
obtained via the polynomial method~\cite{Dvir}.

%\mnote{We should explicitly prove that a non-trivial linear affine invariant
%code has weight at least two.}

\begin{prop}
\label{prop:nikodym}
If $\calL \subset \{\F_q^m \to \F_q\}$ is the lift of some univariate linear
affine-invariant family $\calf \subsetneq \{\F_q \to \F_q\}$, and $N \subseteq \F_q^m$
is a Nikodym set, then $|N| \ge \dim \calL$.
\end{prop}

\begin{proof}
Suppose for sake of contradiction that $|N| < \dim \calL$. Then there exists
nonzero $f \in \calL$ such that $f|_N \equiv 0$. Let $x \in \F_q^m$. Then there
is $y \in \F_q^m$ such that $x + ty \in N$ for every $t \in \F_q \setminus \{0\}$.
Define $g(t) = f(x + ty)$. By definition of $\calL$, we have $g \in \calf$, and moreover
$\calf$ is a nontrivial, so by Proposition~\ref{prop:distance-simple}, either $g = 0$ or $\wt(g) \ge 2$. But $g(t) = 0$ for
every $t \ne 0$, hence $g = 0$, and in particular $f(x) = g(0) = 0$.
Since $x$ was arbitrary, this shows that $f$ is identically zero, a contradiction.
\end{proof}

%Letting $\calL$ be the code given by Theorem~\ref{thm:two}, i.e.\ the
%family of $f$ taking values in $\F_p$
%whose restrictions to lines are polynomials of
%degree at most $q-2$, we immediately
%get the following, which proves Theorem~\ref{thm:nikodym}.
%
%\begin{cor}
%Let $q$ be of constant characteristic $p$.
%If $N \subseteq \F_q^m$ is a Nikodym set,
%then $|N| \ge q^m (1 - o(1))$.
%\end{cor}

We are now ready to prove Theorem~\ref{thm:nikodym}.

\begin{proof}[Proof of Theorem~\ref{thm:nikodym}]
Follows immediately by applying Proposition~\ref{prop:nikodym} to
the code $\calL$ obtained from Theorem~\ref{thm:two}, i.e. the family of $f$
taking values in $\F_p$ whose restrictions to lines are polynomials of degree
at most $q-2$.
\end{proof}

For comparison, the bound obtained by the polynomial method is
${m+q-2 \choose m} \approx q^m/m!$, which can be improved to
$q^m/2^m$ using the method of multiplicities.
Other work on finite field Nikodym sets by Li~\cite{Li} as well as
Feng, Li, and Shen~\cite{FLS} obtain lower bounds that
beat the standard polynomial method bound for $m=2$. In particular, \cite{FLS}
obtains a bound of $q^2 - q^{3/2} - q$, which is actually better than our
bound for two dimensions, which is $q^2 - O(q^{\log_2 3/4})$ for characteristic two.
Moreover, their bound applies to $q$ of any characteristic. However, our bounds
are the best known and the only ones achieving $q^m (1 - o(1))$ for $m \ge 3$.

\section{General investigation of lifting}
\label{sec:distance}

The codes of the previous section simply picked some
basic codes and lifted them to derive long codes of reasonable distance
and interesting local testability and decodability.
To go beyond this setting, we feel it is important to pick basic
codes of possibly high distance and then lift them, and this could
improve the performance of such codes. As may be observed
from the previous section most of the work needed to analyze
lifted codes is devoted to determining their dimension, and this
can be a function of the exact code chosen. Features such
as distance, decodability, and testability seem to follow more
generically. In this section, we examine the simplest of
these properties, namely the distance of the lifted code and
prove some basic facts.

\begin{theorem}
\label{thm:dist}
Let $\calf \subseteq \{\F_Q^t \to \F_q\}$ and
$\calL = \Lift_m(\calf)$ for some $m\geq t$. We have the
following:
\begin{enumerate}
\item \label{dist-one} $\delta(\calL) \leq \delta(\calf)$.
\item \label{dist-two} $\delta(\calL) \geq \delta(\calf) - Q^{-t}$.
\item \label{dist-three} If $Q \in \{2,3\}$ and $\delta(\calf) > Q^{-t}$
then $\delta(\calL) \geq \delta(\calf)$.
\end{enumerate}
\end{theorem}

\subsection{Proof of Theorem~\ref{thm:dist}}
\label{ssec:dist-proof}

We divide the proof of Theorem~\ref{thm:dist} into several parts.
We start by proving that distance does not increase under
lifting (Theorem~\ref{thm:dist}, Part~\ref{dist-one}).

\begin{lemma}
Let $\cF \subseteq \{\F_Q^t \to \F_q\}$ be a linear affine-invariant
code with
lift $\calL = \lift_m(\cF)$. Then $\delta(\calL) \le \delta(\cF)$.
\end{lemma}

\begin{proof}
By induction, it suffices to show the assertion for the case $m=t+1$.
Let $f \in \cF$ and let $\delta = \delta(f,0)$.
Let $\vec x = \langle x_1,\ldots,x_t \rangle$.
Now consider the function $g(\vec x,y) = f(\vec x)$.
Clearly we have $\delta(g,0) = \delta$. We claim
that $g \in \calL$, which completes the proof.
To do so we will show that $g|_H \in \calf$ for every
$t$-dimensional affine subspace $H \subseteq \F_Q^m$.
Fix such a subspace $H$ and let $A:\F_Q^t \to \F_Q^m$
be an affine map whose image is $H$ (such a map does exist).
Note that $g|_H(\vec z) = f(A(\vec z)_1,\ldots,A(\vec z)_t)$.
Thus if we let $A':\F_Q^t \to \F_Q^t$ be the affine map
given by the projection of $A$ to its first $t$ coordinates,
we have that $g|_H = f \circ A'$. By Theorem~\ref{thm:invariance}
$f \circ A' \in \calf$ and so we have $g \in \calf$ as claimed.
(Note that we need to use Theorem~\ref{thm:invariance} since
$A'$ need not be an affine permutation but it is an affine
transformation.)
\end{proof}

Next we prove Part~\ref{dist-three} of Theorem~\ref{thm:dist}
which asserts that the distance of non-trivial binary codes
does not decrease with lifting.

\begin{lemma}
\label{lem:binarylift}
If $\cF \subseteq \{\F_2^t \to \F_2\}$ has distance $\delta(\cF) > \frac{1}{2^t}$, then
$\delta(\lift_m(\cF)) \ge \delta(\cF)$ for all $m \ge t$.
\end{lemma}

We prove the above lemma by stating and proving
the following stronger lemma first.

\begin{lemma}\label{lemma:binarylift}
For all $m \ge 2$, if $\delta > \frac{1}{2^{m-1}}$ and $f : \F_2^m \to \F_2$ such that
$0 < \Pr_{x \in \F_2^m} [f(x) \ne 0] < \delta$,
then there exists an $(m-1)$-dimensional affine subspace $H \subsetneq \F_2^m$
such that $0 < \Pr_{x \in H} [f(x) \ne 0] < \delta$.
\end{lemma}

\begin{proof}
We proceed by induction on $m$. The base case $m=2$
is straightforward to verify.

Now suppose $m > 2$ and our assertion holds for $m-1$.
Let $H_0, H_1$ be the affine subspaces given by $x_m = 0$ and $x_m = 1$
respectively.
Let $\delta_0, \delta_1$ denote $\delta(f|_{H_0},0), \delta(f|_{H_1},0)$ respectively.
Note that $\delta > \delta(f,0) = (\delta_0 + \delta_1)/2$.
If both $\delta_0, \delta_1 > 0$, then by averaging we have $0 < \delta_i < \delta$
and so $H = H_i$ does the job. Otherwise, suppose w.l.o.g. that $\delta_1 = 0$.
Note that $0 < \delta_0 < 2\delta$ and $2\delta > \frac{1}{2^{m-2}}$.
Thus, by the induction hypothesis, there exists an $(m-2$)-dimensional affine
subspace $H'_0 \subset H_0$ such that $0 < \delta(f|_{H'},0) < 2\delta$.
Let $H'_1 = \{(a_1,\ldots,a_{m-1},1) \in \F_2^m \mid (a_1,\ldots,a_{m-1},0) \in H'_0\}$
be the translate of $H'_0$ in $H_1$, and note that
$\delta(f|_{H'_1},0) = 0$. Let $H = H'_0 \cup H'_1$. Then $H$ is
an $(m-1)$-dimensional subspace of $\F_2^m$ such that
$0 < \delta(f|_H,0) = (\delta(f|_{H'_0},0) + \delta(f|_{H'_1},0))/2 < \delta$.
\end{proof}

\begin{proof}[Proof of Lemma~\ref{lem:binarylift}]
We prove the lemma by induction on $m-t$. Indeed the inductive step
is straightforward since $\Lift_m(\calf) = \Lift_m(\Lift_{m-1}(\calf))$
and by induction both lifts on the RHS have smaller value of $m-t$
and so the distance does not reduce in either step. The main case
is thus the base case with $m = t+1$.

Suppose $f \in \lift_m(\cF) \subsetneq \{\F_2^{m} \to \F_2\}$ such that
$0 < \delta(f,0) < \delta(\calf)$.
By~Lemma~\ref{lemma:binarylift}, there exists an
$(m-1)$-dimensional affine subspace
$H \subset \F_2^{m}$ such that
$0 < \delta(f|_H,0) \le \delta(f,0) < \delta$, contradicting the fact that
$f|_H \in \calf$.
\end{proof}

A similar approach works for $q=3$, thus we have the following.

\begin{lemma}\label{lem:3arylift}
If $\calf \subseteq \{\F_3^t \to \F_3\}$ has distance $\delta(\cF) > \frac{1}{3^t}$, then
$\delta(\lift_m(\cF)) \ge \delta(\cF)$ for all $m \ge t$.
\end{lemma}

Again, we prove this by stating and proving the following analogue of
Lemma~\ref{lemma:binarylift}.

\begin{lemma}
\label{lemma:3arylift}
For all $m \ge 2$, if $f : \F_3^m \to \F_3$ such that
$\delta(f,0) \ge \frac{1}{3^{m-1}}$, then there exists an
$(m-1)$-dimensional affine subspace $H \subset \F_3^m$ such that
$0 < \delta(f|_H,0) \le \delta(f,0)$.
\end{lemma}

\begin{proof}
Let $\delta = \delta(f,0)$.
We proceed by induction on $m$. For the base case $m=2$,
$\delta \ge \frac{1}{3}$. Suppose $f = f(x,y)$ and consider $f|_{y=i}$
for $i \in \F_3$. If $f|_{y=i}$ is not identically zero for all $i \in \F_3$,
then by averaging there is some $i \in \F_3$ for which
$0 < \Pr_{x \in \F_3}[f(x,i) \ne 0] \le \delta$. Otherwise, w.l.o.g.
suppose $f|_{y=2} \equiv 0$. Further, w.l.o.g. suppose
$f|_{y=0} \not\equiv 0$ and $f(0,0) \ne 0$. Now, if $\delta \ge \frac23$, then
the line $H = \{(x,y) \in \F_3^2 \mid x=0\}$ does the job, since
$0 < \Pr_{y \in \F_3} [f(0,y) \ne 0] \le \frac23 \le \delta$. If $\delta < \frac23$,
then there must exist some $a,b \in \F_3$ and $c \in \{0,1\}$ such that
$f(a,c) \ne 0$ and $f(b,1-c) = 0$. Then the line $H = \{(a,c), (b,1-c), (2b-a,2)\}$
does the job, since $0 < \Pr_{(x,y) \in H} [f(x,y) \ne 0] = \frac13 \le \delta$.

Now suppose $m > 2$ and the assertion holds for $m-1$.
For $i \in \F_3$, let $H_i$ be the hyperplane cut out by $x_m = i$ and
let $\delta_i = \delta(f|_{H_i},0)$. Then $\delta_1 + \delta_2 + \delta_3 = 3\delta$.
If $\delta_i > 0$ for all $i \in \F_3$, then by simple averaging for some $i \in \F_3$
we have $0 < \delta_i \le \delta$, so assume w.l.o.g. $\delta_2 = 0$
and $\delta_0 \ge \delta_1$. First suppose $\delta_0 \ge \frac{1}{3^{m-2}}$.
Then, by the inductive hypothesis, there exists an $(m-2)$-dimensional
affine subspace $H \subset H_0$ such that $0 < \delta(f|_{H},0) \le \delta_1$.
Let $H^{(0)}$ be defined by the linear equations $\sum_{i=1}^m a_i x_i - a_0 = 0$
and $x_m = 0$ for some $\langle a_0,\ldots,a_m \rangle \in \F_3^{m+1}$.
For each $i,j \in \F_3$, let $H^{(i)}+j \subset H_1$ denote the affine subspace defined by
$\sum_{i=1}^m a_ix_i - a_0 = j$ and $x_m = i$. By averaging, for some $i \in \F_3$,
$\delta(f|_{H^{(1)}+i},0) \le \delta_2$.
Take $H = H^{(0)} \cup (H^{(1)}+i) \cup (H^{(2)}+2i)$. Then
$0 < \delta(f|_H,0) \le \delta$. Otherwise, suppose
$\frac{1}{3^{m-2}} > \delta_0$, so $\delta_0, \delta_1 \le \frac{2}{3^{m-1}}$.
There exists $H^{(0)} \subset H_0$ be an $(m-2)$-dimensional
affine subspace such that $\delta(f|_{H^{(0)}},0) = \frac{1}{3^{m-1}}$.
To see this, let $a,b \in H_0$ such that $f(a), f(b)$ are nonzero, and suppose
$a$ and $b$ differ in the $k$-th coordinate. Then take $H^{(0)}$ defined by
$x_k = a_k$ and $x_m = 0$. Again, for $i,j \in \F_3$ let $H^{(j)}+i$ be the
$(m-2)$-dimensional affine subspace defined by $x_k = a_k+i$ and $x_m = j$.
Since $\delta_2 \le \frac{2}{3^{m-2}}$, there is $i \in \F_3$ such that
$f|_{H^{(1)}+i} \equiv 0$. Then, taking
$H = H^{(0)} \cup (H^{(1)}+i) \cup (H^{(2)}+2i)$, we have
$0 < \delta(f|_H,0) = \frac{1}{3^{m-1}} \le \delta$.
\end{proof}

\begin{proof}
[Proof of Lemma~\ref{lem:3arylift}]
We prove the lemma by induction on $m-t$. The inductive step is
straightforward since $\lift_m(\calf) = \lift_m(\lift_{m-1}(\calf))$ and by induction
both lifts on the RHS have smaller value of $m-t$ and so the distance does not
reduce in either step. The main case is thus the base case with $m=t+1$.

Suppose $f \in \lift_m(\calf) \subsetneq \{\F_3^m \to \F_3\}$ such that
$0 < \delta(f,0) < \delta(\calf)$. If $\delta(f,0) \ge \frac{1}{3^{m-1}}$, then,
by Lemma~\ref{lemma:3arylift}, there exists
an $(m-1)$-dimensional affine subspace $H \subset \F_3^m$ such that
$0 < \delta(f|_H,0) \le \delta(f,0) < \delta(\calf)$, contradicting the fact that
$f|_H \in \calf$.
If $\delta(f,0) < \frac{1}{3^{m-1}}$, then there are at most two points $a,b \in \F_3^m$
such that $f(a),f(b)$ are nonzero. Let $i \in [m]$ such that $a_i \ne b_i$ and
let $H$ be the hyperplane defined by $x_i = a_i$. Then $f|_H$ is nonzero only on $a$,
so $0 < \delta(f|_H) = \frac{1}{3^{m-1}} < \delta(\calf)$, again contradicting the fact
that $f|_H \in \calf$.
\end{proof}

For general $q > 3$, we have the following.

\begin{lemma}\label{lem:genarylift}
If $\cF \subseteq \{\F_Q^t \to \F_q\}$ has distance $\delta(\cF) = \delta$,
then $\delta(\lift_m(\cF)) > \delta - \frac{1-\delta}{Q^t-1}$.
\end{lemma}
\begin{proof}
Fix a non-zero $f \in \lift_m(\cF)$ and let $\tau = \delta(f,0)$.
Fix $a \in \F_Q^m$ such that $f(a) \ne 0$.
Now let $A$ be a $t$-dimensional affine subspace containing $a$ 
chosen uniformly at random from all such subspaces.
Let $X(A) = |\{x \in A \mid f(x) \ne 0\}|$ be the random variable 
denoting the number of non-zero points of $f$ on $A$.
Since $A$ samples every point of $\F_Q^n - \{a\}$ uniformly,
we have
$$
\E_A[X(A)] = 1 + \frac{\tau Q^m-1}{Q^m-1}(Q^t-1) < 1 +
\tau(Q^t-1).$$
Therefore there must exist a $t$-dimensional subspace $A$
containing $a$ with $X(A) < \tau(Q^t - 1) + 1$. Since $f|_A$
is a non-zero function in $\calf$, we have $\tau(Q^t - 1) + 1 \geq 
\delta Q^t$ and thus we conclude that $\tau \geq \delta - 
\frac{1 - \delta}{Q^t-1}$. In other words every non-zero 
function in $\calf$ is non-zero on $\delta - \frac{1 - \delta}{Q^t - 1}$
fraction of the points, as asserted.
\end{proof}

Finally we mention examples which show that, in some senses
the gaps in
Theorem~\ref{thm:dist}, Parts~\ref{dist-two}~and~\ref{dist-three}
are inherent.

First note that if $\calf = \{F_Q^t \to \F_q\}$ then
$\Lift_m(\calf) = \{\F_Q^m \to \F_q\}$ whose distance is
$Q^{-m}$, and so some loss in the distance is inherent in
Part~\ref{dist-two} of Theorem~\ref{thm:dist}.
However, one could hope that if $\calf \subsetneq
\{F_Q^t \to \F_q\}$ then its distance is preserved by
lifting (as in Part~\ref{dist-three} of Theorem~\ref{thm:dist}).
Unfortunately (actually fortunately, since this is where
the rate improvement of codes in Theorem~\ref{thm:one} comes from)
even this hope is not true.
If one takes $\calf$ to be the binary code with degree set being all
weight one integers, then its lift contains all the
weight one integers as well as some integers of
weight greater than one. The code
consisting of only weight one integers in its degree set
has distance exactly $1/2$ while codes that have rate greater
than these must have distance strictly smaller than $1/2$ (by
the Plotkin bound). This suggests that distances can reduce
under lifts. A search reveals that the code $\calf \subseteq
\{\F_4 \to \F_2\}$ with degree set $\Deg(\calf) = \{0,1,2\}$
has distance $1/2$ while its lift $\calL = \Lift_2(\calf)$
has distance $3/8$.

\section*{Acknowledgments}

We would like to thank Sergey Yekhanin for introducing us to
the projective space codes which led to the parameter settings of
Section~\ref{ssec:two}.
We would like to thank Elad Haramaty for clarifying discussions
on the relationship between
the definition of lifting in \cite{BGMSS11-ECCC} and
in this work.

\bibliographystyle{plain}
\bibliography{coding}

\appendix

\section{Equivalence of invariance under affine transformations
and permutations}
\label{app:perm}

In their work initiating the study of the testability 
of affine-invariant properties
(codes), Kaufman and Sudan~\cite{KS08-ECCC} studied properties
closed under general affine transformations and not just permutations.
While affine transformations are nicer to work with when available,
they are not mathematical elegant (they don't form a group under
composition). Furthermore in the case of codes they also do not
preserve the code - they only show that every codeword stays in
the code after the transformation. Among other negative features
affine transformations do not even preserve the weight of non-zero
codewords, which can lead to some rude surprises.
Here we patch the gap by showing that families closed under
affine permutations are also closed under affine transformations.
So one can assume the latter, without restricting the class
of properties under consideration.
We note that such a statement was proved in \cite{BGMSS11-ECCC} for the
case of univariate functions. Unfortunately their proof does
not extend to the multivariate setting and forces us to rework
many steps from \cite{KS08-ECCC}.

\begin{theorem}\label{thm:invariance}
If $\cF \subseteq \{\F_Q^m \to \F_q\}$ is an $\F_q$-linear code
invariant under affine permutations, then
$\cF$ is invariant under all affine transformations.
\end{theorem}

The central lemma (Lemma~\ref{lem:split})
that we prove is that every non-trivial function
can be split into more basic ones.
This leads to a proof of Theorem~\ref{thm:invariance}
fairly easily.

We first start with the notion of a basic function.
For $Q = q^n$, let $\Tr:\F_Q \to \F_q$ denote the {\em trace} function
$\Tr(x) = x + x^q + \cdots + x^{q^{n-1}}$.
We say that $f:\F_Q^m \to \F_q$ is a {\em basic} function if
$f(\vec x) = \Tr(\lambda \vec x^{\vec d})$ for some
$\vec d \in \{0,\ldots,Q-1\}^m$.
For $\calf \subseteq \{\F_Q^m \to \F_q\}$ and $f \in \calf$
we say $f$ can be {\em split} (in $\calf$) if there exist
functions $g$ and $h$ such that $f = g+h$ and
$\supp(g), \supp(h) \subsetneq \supp(f)$.

\begin{lemma}
\label{lem:split}
If $\cF \subseteq \{\F_Q^m \to \F_q\}$ is an $\F_q$-linear code
invariant under affine permutations, then for
every function $f \in \calf$, $f$ is either basic or $f$ can be
split.
\end{lemma}

We first prove Theorem~\ref{thm:invariance} from Lemma~\ref{lem:split}.

\begin{proof}[Proof of Theorem~\ref{thm:invariance}]
First we assert that it suffices to prove that for every function
$f \in \calf$ the function $\tilde{f} =
f(x_1,\ldots,x_{m-1},0)$ is also in $\calf$.
To see this, consider $f \in \calf$ and
$A:\F_Q^m \to \F_Q^m$ which is not
a permutation. Then there exists affine permutations
$B,C: \F_Q^m \to \F_Q^m$ such that
$A(\vec x) = B(C(\vec x)_1,\ldots,C(\vec x)_r,0,\ldots,0)$
where $r<m$ is the dimension of the image of $A$.
By closure under affine permutations, it follows
$f \circ C \in \calf$. Applying the assertion above $m-r$
times we have that $f'(\vec x) = f\circ C (x_1,\ldots,x_r,0,\ldots,0)$
is also in $\calf$. Finally $f\circ A = f' \circ B$ is also in $\calf$.
So we turn to proving that
for every $f \in \calf$ the function $\tilde{f} =
f(x_1,\ldots,x_{m-1},0)$ is also in $\calf$.

Let $f(\vec x) = \sum_{\vec d} c_{\vec d} \vec x^{\vec d}$.  Notice
$\tilde{f}(\vec x) = \sum_{\vec d | d_m = 0} c_{\vec d} \vec x^{\vec d}$.
Writing $f = \tilde{f} + f_1$, we use Lemma~\ref{lem:split} to
split $f$ till we express
it as a sum of basic functions $f = \sum_{i=1}^N b_i$, where each
$b_i$ is a basic function in $\calf$. Note that for every $b_i$,
we have $\supp(b_i) \subseteq \supp(\tilde{f})$ or
$\supp(b_i) \subseteq \supp(f_1)$ (since the trace preserves
$d_m = 0$).
By reordering the $b_i$'s assume the first $M$ $b_i$'s have their
support in the support of $\tilde{f}$. Then
we have $\tilde{f} = \sum_{i=1}^M b_i \in \calf$.
\end{proof}

We thus turn to the proof of Lemma~\ref{lem:split}.
We prove the lemma in a sequence of cases, based on the
kind of monomials that $f$ has in its support.

We say that $\vec d$ and $\vec e$ are equivalent (modulo $q$),
denoted $\vec d \equiv_q \vec e$ if there exists a $j$
such that for every $i$, $d_i = q^j e_i \modstar Q$.
The following proposition is immediate from previous works
(see, for example, \cite{BGMSS11-ECCC}). We include a
proof for completeness.

\begin{proposition}
If every pair $\vec d,\vec e$ in the support of
$f:\F_Q^m \to \F_q$ are equivalent, then $f$ is a basic function.
\end{proposition}

\begin{proof}
We first note that since the $\Tr:\F_Q \to \F_q$ is a
$(Q/q)$-to-one function, we have in particular that for
every $\beta \in \F_q$ there is an $\alpha \in \F_Q$ such
that $\Tr(\alpha) = \beta$.
As an immediate consequence we have that
every function $f:\F_Q^m \to \F_q$
can be expressed $\Tr \circ g$ where $g:\F_Q^m \to \F_Q$.
Finally we note that we can view $g$ as an element of
$\F_Q[\vec x]$, to conclude that $f = \Tr \circ g$ for some
polynomial $g$.

Now fix $f:\F_Q^m \to \F_q$ all of whose monomials are equivalent.
By the above we can express $f = \Tr \circ g$ for some polynomial
$g$. By inspection we can conclude that all monomials in the support
of $g$ are equivalent to the monomials in the support of $f$.
Finally,
using the fact that $\Tr(\alpha \vec x^{\vec d})
= \Tr(\alpha^q \vec x^{q \vec d \modstar Q})$ we can assume that
$g$ is supported on a single monomial and so $f = \Tr(\lambda \vec x^{\vec
d}$ for some $\lambda \in \F_Q$.
\end{proof}

So it suffices to show that every function that contains
non-equivalent degrees in its support can be split. We
first prove that functions with ``non-weakly-equivalent'' monomials
can be split.

We say that $\vec d$ and $\vec e$ are weakly equivalent
if there exists a $j$ such that for every $i$,
$d_i = q^j e_i (\mod Q-1)$.

\begin{lemma}
\label{lem:non-weak-equiv-split}
If $\cF \subseteq \{\F_Q^m \to \F_q\}$ is an $\F_q$-linear code
invariant under affine permutations and $f \in \calf$ contains
a pair of non-weakly equivalent monomials in its support, then
$f$ can be split.
\end{lemma}

\begin{proof}
Let $\vec d$ and $\vec e$ be two non weakly-equivalent
monomials in the support of $f$.
Fix $j$ and
consider the function $f_j(\vec x) = \sum_{\vec a \in (\F^*_Q)^m}
\prod a_i^{- q^j d_i} f(a_1 x_1,\ldots,a_m x_m)$.
We claim that (1) the support of $f_j$ is a subset of the support
of $f$, (2) $q^j \vec d$ is in the support of $f_j$,
(3) $\vec f$ is in the support of $f_j$ only if for every $i$
$f_i = q^j d_i (\mod Q-1)$ and in particular
(4) $\vec e$ is not in the support of $f_j$.

Now let $b= b(\vec d)$ be the smallest positive integer such
that $q^b d_i = d_i \modstar Q$ for every $i$.
Now consider the function $g = \sum_{j=0}^{b-1} f_j$.
We have that $g \in \calf$ since it is an $\F_q$-linear
combination of linear transforms of functions in $\calf$.
By the claims about the $f_j$'s we also have
that $\vec d$ is in the support of $g$, the support of $g$ is
contained in the support of $f$ and $\vec e$ is not in
the support of $f$. Expressing $f = g + (f - g)$ we
now have that $f$ can be split.
\end{proof}

The remaining cases are those where some of coordinates
of $\vec d$ are zero or $Q-1$ for every $\vec d$ in the support
of $f$. We deal with a special
case of such functions next.

\begin{lemma}
\label{lem:zeroQ-1}
Let $\calf$ be a linear
affine-invariant code. Let $f \in \calf$ be given
by $f(\vec x,\vec y) = \Tr(\vec y^{\vec d} p(\vec x))$ where
every variable in $p(\vec x)$ has degree in $\{0,Q-1\}$ in every
monomial, and $\vec d$ is arbitrary.
Further, let degree of $p(\vec x)$ be $a(Q-1)$. Then for every
$0 \leq b \leq a$ and for every $\lambda \in \F_Q$,
the function $(x_1\cdots x_b)^{Q-1} \Tr(\lambda \vec y^{\vec d}) \in
\calf$.
\end{lemma}

Note that in particular the lemma above implies that such $f$'s
can be split into basic functions.

\begin{proof}
We prove the lemma by a triple induction, first on $a$,
then on $b$, and then on the number of monomials in $p$.
The base case is $a = 0$ and that is trivial.
So we consider general $a > 0$.

First we consider the case $b < a$.
Assume w.l.o.g. that the monomial $(x_1\cdots x_a)^{Q-1}$
is in the support of $p$
and write $p = p_0 + x_1^{Q-1} p_1$ where $p_0,p_1$ do
not depend on $x_1$. Note that $p_1 \ne 0$ and
$\deg(p_1) = (a-1)(Q-1)$. We will prove that $- \Tr(\vec y^{\vec d} p_1(\vec x)) \in \calf$
and this will enable us to apply the inductive hypothesis
to $p_1$. Let $g(\vec x, \vec y) = \sum_{\beta \in \F_Q}
f(x_1+\beta,x_2,\ldots,x_m,\vec y)$.
By construction $g \in \calf$.
By linearity of the Trace we have
$$g = \Tr\left(\vec y^{\vec d}
  \left(\sum_{\beta \in \F_Q} p_0 + (x_1+\beta)^{Q-1}p_1 \right)\right)
  = \Tr(\vec y^{\vec d} (-p_1(\vec x))),$$
where the second equality follows from the fact that
$\sum_{\beta \in \F_Q} (z+\beta)^{Q-1} = -1$.
Thus we can now use induction to claim
$(x_1\ldots x_b)^{Q-1} \Tr(\lambda \vec y^{\vec d}) \in \calf$.

Finally we consider the case $b=a$.
Now note that since the case $b<a$ is known, we can assume
w.l.o.g that $p$ is homogenous (else we can subtract off the
lower degree terms). Now if $a = m$ there is nothing to be
proved since $p$ is just a single monomial. So assume $a < m$.
Also if $p$ has only one monomial then there is nothing to
be proved, so assume $p$ has at least two monomials. In particular
assume $p$ is supported on some monomial that depends on $x_1$
and some monomial that does not depend on $x_1$. Furthermore,
assume w.l.o.g. that a monomial depending on $x_1$ does not
depend on $x_2$. Write
$p = x_1^{Q-1} p_1 + x_2^{Q-1} p_2 + (x_1 x_2)^{Q-1} p_3 + p_4$
where the $p_i$'s don't depend on $x_1$ or $x_2$.
By assumption on the monomials of $p$ we have that $p_1 \ne 0$ and
at least one of $p_2,p_3,p_4 \ne 0$.
Now consider the affine
transform $A$ that sends $x_1$ to $x_1 + x_2$ and preserves all other
$x_i$'s.
We have
$g = f \circ A = \Tr\left(\vec y^{\vec d} (x_1^{Q-1} p_1 +
x_2^{Q-1}(p_1 + p_2) + (x_1x_2)^{Q-1}p_3 +
p_4 + r)\right)$
where the $x_1$-degree of every monomial in $r$ is in
$\{1,\ldots,Q-2\}$.
Now consider $g'(\vec x,\vec y) = \sum_{\alpha \in \F_Q^*} g(\alpha
x_1,x_2,\ldots,x_m,\vec y)$. The terms of $r$ vanish in $g'$
leaving $g' =
-(f \circ A - r) =\Tr\left(\vec y^{\vec d} \left(- x_1^{Q-1} p_1 -
x_2^{Q-1}(p_1 + p_2) - (x_1x_2)^{Q-1} p_3 -
p_4\right)\right)$.
Finally we consider the function $\tilde{g} = f + g' =
\Tr(\vec y^{\vec d} (- x_2^{Q-1} p_1))$
which is a function in $\calf$ of degree $a(Q-1)$ supported on a smaller
number of monomials than $f$, so by applying the inductive hypothesis
to $\tilde{g}$ we have that $\calf$ contains the monomial
$(x_1\cdots x_a)^{Q-1}$.
\end{proof}

The following lemma converts the above into
the final piece needed to prove Lemma~\ref{lem:split}.

\begin{lemma}
\label{lem:weak-equiv-split}
If $\cF \subseteq \{\F_Q^m \to \F_q\}$ is an $\F_q$-linear code
invariant under affine permutations and all monomials in
$f \in \calf$ are weakly equivalent, then $f$ can be split.
\end{lemma}

\begin{proof}
First we describe the structure of a function $f:\F_Q^m \to \F_q$
that consists only of weakly equivalent monomials. First we
note that the $m$ variables can be separated into those in which
every monomial has degree in $\{1,\ldots,Q-2\}$ and those in which
every monomial has degree in $\{0,Q-1\}$ (since every monomial is
weakly equivalent). Let us denote by $\vec x$ the variables in which
the monomials of $f$ have degree in $\{0,Q-1\}$ and $\vec y$ be the
remaining monomials. Now consider some monomial of the form
$M = c \vec x^{\vec e} \vec y^{\vec d}$ in $f$. Since
$f$ maps to $\F_q$ we must have that the coefficient of $(\vec x^{\vec
e} \vec y^{\vec d})^{q^j}$ is $c^{q^j}$.
Furthermore, we have every other monomial $M'$ in the support of
$f$ is of the form $c' \vec y^{q^j \vec d} \vec x^{\vec e'}$.
Thus $f$ can be written as $\Tr(\vec y^{\vec d} p(\vec x))$ where
$p(x_1,\ldots,x_m) = \tilde{p}(x_1^{Q-1},\ldots,x_m^{Q-1})$.
But, by Lemma~\ref{lem:zeroQ-1}, such an $f$ can be split.
\end{proof}

\begin{proof}[Proof of Lemma~\ref{lem:split}]
If $f$ contains a pair of non-weakly equivalent monomials
then $f$ can be split by Lemma~\ref{lem:non-weak-equiv-split}.
If not, then $f$ is either basic or, by Lemma~\ref{lem:weak-equiv-split}
is can be split.
\end{proof}

We also prove an easy consequence of Lemma~\ref{lem:split}.

\begin{lemma}
\label{lem:monomial}
Let $\calf \subseteq \{\F_Q^m \to \F_q\}$ be affine invariant.
If $\vec d \in \Deg(\calf)$, then $\Tr(\lambda \vec x^{\vec d}) \in \cF$
for all $\lambda \in \F_Q$.
\end{lemma}

\begin{proof}
We first claim that Lemma~\ref{lem:split} implies
that there exists $\beta \in \F_Q$ such that $\Tr(\beta \vec x^{\vec d})$ is
a non-zero function in $\calf$.
To verify this, consider a ``minimal'' function
(supported on fewest monomials) $f \in \calf$ with $\vec d \in \supp(f)$. Since $f$ can't be split in $\calf$ (by minimality), by
Lemma~\ref{lem:split} $f$ must be basic and so equals (by definition
of being basic) $\Tr(\beta \vec x^{\vec d})$.

Now let $b = b(\vec d)$ be the smallest positive integer such
that $q^b {\vec d} \modstar Q = \vec d$. If $Q = q^n$, note that
$b$ divides $n$ and so one can write $\Tr:\F_Q \to \F_q$
as $\Tr_1 \circ \Tr_2$ where $\Tr_1:\F_{q^b} \to \F_q$
is the function $\Tr_1(z) = z + z^q + \cdots + z^{q^{b-1}}$
and $\Tr_2:\F_Q \to \F_{q^b}$ is the function
$\Tr_2(z) = z + z^{q^b} + \cdots + z^{Q/q^b}$.
(Both $\Tr_1$ and $\Tr_2$ are trace functions mapping the
domain to the range.)
It follows that $\Tr(\beta \vec x^{\vec d}) =
\Tr_1(\Tr_2(\beta) \vec x^{\vec d})$.

We first claim that $\Tr_1(\tau \vec x^{\vec d}) \in \calf$
for every $\tau \in \F_{q^b}$.
Let $S =
\{ \sum_{\alpha \in (\F_Q^*)^m} a_{\alpha} \cdot \alpha^{\vec d} \mid a_{\alpha} \in \F_q\}$.
We note that by linearity and affine-invariance of $\calf$,
we have that $\Tr_1(\Tr_2(\beta) \cdot \eta \vec x^{\vec d}) \in\calf$
for every $\eta \in S$. By definition $S$ is closed under
addition and multiplication and so is a subfield of
$\F_Q$. In fact, since every $\eta \in S$ satisfies
$\eta^{q^b} = \eta$ (which follows from the fact that
$\alpha^{\vec d} = \alpha^{q^b \vec d}$), we have that
$S \subseteq \F_{q^b}$. It remains to show $S = \F_{q^b}$.
Suppose it is a strict subfield of size $q^c$ for $c < b$.
Consider $\gamma^{d_i}$ for $\gamma \in \F_Q$
and $i \in [m]$. Since $\gamma^{d_i} \in S$, we have that
$\gamma^{d_i q^c} = \gamma^{d_i}$ for every $\gamma \in \F_Q$
and so we get $x_i^{q^c d_i} = x_i \mod (x_i^Q - x_i)$.
We conclude that $\vec x^{q^c \vec d} = \vec x^{\vec d} \pmod{\vec x^Q - \vec
x}$ which contradicts the minimality of $b = b(\vec d)$. We conclude
that $S = \F_{q^b}$. Since $\Tr_2(\beta) \in \F_{q^b}^*$, we
conclude that the set of coefficients
$\tau$ such that $\Tr_1(\tau \vec x^{\vec d}) \in \calf$
is all of $\F_{q^b}$ as desired.

Finally consider any $\lambda \in \F_Q$.
since $\Tr_2(\lambda) \in \F_{q^b}$, we have that
$\Tr_1(\Tr_2(\lambda) \vec x^{\vec d}) \in \calf$ (from the
previous paragraph), and
so
$\Tr(\lambda \vec x^{\vec d}) =
\Tr_1(\Tr_2(\lambda) \vec x^{\vec d}) \in \calf$
\end{proof}

\iffalse{
\subsection{Degree sets of univariate vs. multivariate families}

We start by noticing that every affine-invariant family
$\calf \subseteq \{\F_Q^m \to \F_q\}$ corresponds to a
unique affine-invariant family $\calg \subseteq \{\F_{Q^m} \to \F_q\}$.

We say that $\phi:\F_{Q^m} \to \F_Q^m$ is an {\em isomorphism}
if it is a one-to-one function and an $\F_Q$-linear map.

For $\calf \subseteq \{\F_Q^m \to \F_q\}$
let $\calf \circ \phi \subseteq \{\F_{Q^m} \to \F_q\}$
be given by $\calf \circ \phi = \{f \circ \phi | f \in \calf\}$.
We have the following.

\begin{proposition}
Let $\calf \subseteq \{\F_Q^m \to \F_q\}$ by a linear affine-invariant
family, and let $\phi,\psi:\F_{Q^m} \to \F_Q^m\}$ be isomorphisms.
Then we have
\begin{enumerate}
\item $\calf \circ \phi = \calf \circ \psi$.
\item $\calf \circ \phi$ is a linear affine-invariant family.
\end{enumerate}
\end{proposition}

}\fi

\section{Coefficients of multinomial expansions modulo a prime}
\label{app:lucas}

For an integer $d$, let $d^{(i)}$ denote the $i$th digit in the $p$-ary
expansion of $d$, so that $d = \sum_{i=1}^\infty d^{(i)}p^i$.
Let $\equiv_p$ denote equivalence modulo $p$.
The following is a well known theorem of Lucas.

\begin{theorem}[Lucas' theorem]
If $d,e \in \Z$, then ${d \choose e} \equiv_p \prod_i {d_i \choose e_i}$.
\end{theorem}

In particular, ${d \choose e} \not\equiv 0 \pmod{p}$ if and only if $e \le_p d$,
so we have $(x+y)^d \equiv_p \sum_{e \le_p d} x^ey^{d-e}$.
More generally, we would like to know when ${\vec d \choose \vec E}$ vanishes
modulo~$p$. To this end, we use the following claim.

\begin{lemma}
\label{lemma:genlucas}
If $d \in \Z$ and $\vec e \in \Z^t$, then ${d \choose \vec e} \not\equiv_p 0$
only if $\vec e \le_p d$.
More generally, if $\vec d \in \Z^m$ and $\vec E \in \Z^{m \times t}$, then
${\vec d \choose \vec E} \not\equiv_p 0$ only if $\vec E \le_p \vec d$.
\end{lemma}
\begin{proof}
We have ${d \choose \vec e} = \prod_{i=1}^{t-1} {d - \sum_{j=1}^{i-1} e_j \choose e_i}$.
For this to be nonzero modulo~$p$, by Lucas' theorem we have
$e_i \le_p d-\sum_{j=1}^{i-1} e_j$, from which it follows that $\vec e \le_p d$.
The more general statement then follows immediately from definition.
\end{proof}

\begin{lemma}
\label{lemma:expansion}
Let $\vec A \in \Z^{m \times t}$ and let $\vec d, \vec b \in \Z^m$ and $\vec x \in \Z^t$.
Then
$$
(\vec A \vec x + \vec b)^{\vec d} =
\sum_{\vec E} {\vec d \choose \vec E}
\vec A^{\vec E} \vec x^{\Sigma(\vec E)} \vec b^{\vec d - \Sigma(\vec E^{\top})}.
$$
\end{lemma}
\begin{proof}
For matrices $\vec A, \vec E$, let $a_{ij}, e_{ij}$ denote their entries respectively.
The lemma follows by straightforward calculation. We have
\begin{eqnarray*}
(\vec A \vec x + \vec b)^{\vec d}
&=&
\prod_{i=1}^m \left( \sum_{j=1}^t a_{ij}x_j + b_i \right)^{d_i} \\
&=&
\prod_{i=1}^m
\left( \sum_{e_{i1},\ldots,e_{it}} 
{d_i \choose \langle e_{i1},\ldots,e_{it} \rangle}
\left(
\prod_{j=1}^t a_{ij}^{e_{ij}}x_j^{e_{ij}}
\right) b_i^{d_i - \sum_{j=1}^t e_{ij}}\right) \\
&=&
\sum_{\vec E}
\prod_{i=1}^m
\left(
{d_i \choose \langle e_{i1},\ldots,e_{it} \rangle}
\left(
\prod_{j=1}^t
a_{ij}^{e_{ij}} x_j^{e_{ij}}
\right)b_i^{d_i - \sum_{j=1}^t e_{ij}}
\right) \\
&=&
\sum_{\vec E}
{\vec d \choose \vec E}
\left( \prod_{i=1}^m \prod_{j=1}^t a_{ij}^{e_{ij}} \right)
\left( \prod_{j=1}^t x_j^{\sum_{i=1}^m e_{ij}} \right)
\left( \prod_{i=1}^m b_i^{d_i - \sum_{j=1}^t e_{ij}} \right)
\\
&=&
\sum_{\vec E} {\vec d \choose \vec E}
\vec A^{\vec E} \vec x^{\Sigma(\vec E)} \vec b^{\vec d - \Sigma(\vec E^{\top})}.
\end{eqnarray*}
\end{proof}

\end{document}